\documentclass[a4paper, UKenglish, cleveref, thm-restate, numberwithinsect]{lipics-arxiv}

\hideLIPIcs

\usepackage{hyperref}
\usepackage{mathtools}
\usepackage[dvipsnames]{xcolor}
\usepackage{xspace}
\usepackage{thm-restate}
\usepackage{tikz}
\usetikzlibrary{calc, automata, positioning, arrows, decorations.pathreplacing, petri, topaths}

\newcommand{\ie}{i.e.\@\xspace}

\newcommand{\rawn}[1]{acyclic workflow net#1 with resets}

\newcommand{\RAWN}[1]{Acyclic Workflow Net#1 with Resets}
\newcommand{\rapn}[1]{acyclic Petri net#1 with resets}

\newcommand{\RAPN}[1]{Acyclic Petri Net#1 with Resets}
\newcommand{\zapn}[1]{acyclic Petri net#1 with zero tests}

\newcommand{\class}[1]{\textup{\textsf{#1}}}
\newcommand{\set}[1]{\{#1\}}

\newcommand{\size}[1]{|#1|}

\renewcommand{\vec}[1]{{\bf #1}}
\newcommand{\run}[3]{#1\overset{#2}{\rightarrow}#3}
\newcommand{\pre}[1]{{^\bullet #1}}
\newcommand{\post}[1]{{#1^{\bullet}}}
\newcommand{\reset}[1]{{{#1}^{\circ}}}

\newcommand{\I}{\mathsf{i}}
\newcommand{\Fin}{\mathsf{f}}
\newcommand{\N}{\mathbb{N}}
\newcommand{\Z}{\mathbb{Z}}
\newcommand{\Nw}{\mathbb{N}_{\omega}}
\newcommand{\Zw}{\mathbb{Z}_{\omega}}

\newcommand{\Bb}{\mathcal{B}}
 
\newcommand{\Nn}{\mathcal{N}}
\newcommand{\Pp}{\mathcal{P}}
\newcommand{\Rr}{\mathcal{R}}
\newcommand{\Uu}{\mathcal{U}}
\newcommand{\Ww}{\mathcal{W}}
\newcommand{\Zz}{\mathcal{Z}}
\newcommand{\m}{\vec{m}}

\renewcommand{\phi}{\varphi}

\renewcommand{\sim}{t_\textup{sim}}
\newcommand{\con}{t_\textup{con}}
\newcommand{\pro}{t_\textup{pro}}
\newcommand{\problemx}[3]{
\par\noindent\underline{#1}\par\nobreak\vskip.2\baselineskip
\begingroup\clubpenalty10000\widowpenalty10000
\setbox0\hbox{\bf INPUT:\ }\setbox1\hbox{\bf QUESTION:\ }
\dimen0=\wd0\ifnum\wd1>\dimen0\dimen0=\wd1\fi
\vskip-\parskip\noindent
\hbox to\dimen0{\box0\hfil}\hangindent\dimen0\hangafter1\ignorespaces#2\par
\vskip-\parskip\noindent
\hbox to\dimen0{\box1\hfil}\hangindent\dimen0\hangafter1\ignorespaces#3\par
\endgroup}
\DeclarePairedDelimiter\abs{\lvert}{\rvert}%
\DeclarePairedDelimiter\norm{\lVert}{\rVert}%
\makeatletter
\let\oldabs\abs
\def\abs{\@ifstar{\oldabs}{\oldabs*}}
\let\oldnorm\norm
\def\norm{\@ifstar{\oldnorm}{\oldnorm*}}
\makeatother

\title{Acyclic Petri and Workflow Nets with Resets}

\author{Dmitry Chistikov}{Centre for Discrete Mathematics and its Applications (DIMAP) \&\\ Department of Computer Science, University of Warwick, Coventry, UK}{d.chistikov@warwick.ac.uk}{0000-0001-9055-918X}{Supported in part by the Engineering and Physical Sciences Research Council [EP/X03027X/1].}

\author{Wojciech Czerwi\'{n}ski}{University of Warsaw, Poland}{wczerwin@mimuw.edu.pl}{0000-0002-6169-868X}{Supported by the ERC grant INFSYS, agreement no. 950398.}

\author{Piotr Hofman}{University of Warsaw, Poland}{piotr.hofman@uw.edu.pl}{0000-0001-9866-3723}{Supported by the ERC grant INFSYS, agreement no. 950398.}

\author{Filip Mazowiecki}{University of Warsaw, Poland}{f.mazowiecki@mimuw.edu.pl}{}{Supported by the ERC grant INFSYS, agreement no. 950398.}

\author{Henry Sinclair-Banks}{Centre for Discrete Mathematics and its Applications (DIMAP) \&\\ Department of Computer Science, University of Warwick, Coventry, UK \and \url{http://henry.sinclair-banks.com}}{h.sinclair-banks@warwick.ac.uk}{https://orcid.org/0000-0003-1653-4069}{Supported by EPSRC Standard Research Studentship (DTP), grant number EP/T5179X/1.}

\authorrunning{D. Chistikov, W. Czerwi\'{n}ski, P. Hofman, F. Mazowiecki, and H. Sinclair-Banks}



\nolinenumbers

\acknowledgements{
    We would like to thank Alain Finkel for pointing out an error in the introduction.
    We would also like to thank our anonymous reviewers for their detailed comments.
}

\EventEditors{John Q. Open and Joan R. Access}
\EventNoEds{2}
\EventLongTitle{42nd Conference on Very Important Topics (CVIT 2016)}
\EventShortTitle{CVIT 2016}
\EventAcronym{CVIT}
\EventYear{2016}
\EventDate{December 24--27, 2016}
\EventLocation{Little Whinging, United Kingdom}
\EventLogo{}
\SeriesVolume{42}
\ArticleNo{23}

\begin{document}

\maketitle

\begin{abstract}
    In this paper we propose two new subclasses of Petri nets with resets,
for which the reachability and coverability problems become tractable.
Namely, we add an acyclicity condition that only applies to the consumptions and productions, not the resets.
The first class is acyclic Petri nets with resets, and we show that coverability is \class{PSPACE}-complete for them.
This contrasts the known \class{Ackermann}-hardness for coverability in (not necessarily acyclic) Petri nets with resets.
We prove that the reachability problem remains undecidable for acyclic Petri nets with resets.
The second class concerns workflow nets, a practically motivated and natural subclass of Petri nets.
Here, we show that both coverability and reachability in acyclic workflow nets with resets are \class{PSPACE}-complete.
Without the acyclicity condition, reachability and coverability in workflow nets with resets are known to be equally hard as for Petri nets with resets, that being \class{Ackermann}-hard and undecidable, respectively.
\end{abstract}

\section{Introduction}
\label{sec:introducion}

Petri nets~\cite{Petri75} are among the most fundamental formalisms for modelling processes.
They are defined by a finite set of places and a finite set of transitions. 
A configuration of a Petri net, known as a marking, is a vector of dimension equal to the number of places, with entries equal to the number of tokens in particular places.
Transitions change markings by consuming and producing tokens in places. 
For an example, see~\cref{fig:petri-net}.

\begin{figure}
	\centering
	\begin{tikzpicture}[node distance=1.3cm, >=stealth', bend angle=45, auto]
	    \node [place,tokens=2] (I) [label=below:$\I$] {};
	    \node [place] (p1) [above right = 1cm and 3cm of I,label=below:$p_1$] {};
	    \node [place] (p2) [below right = 1cm and 3cm of I,label=below:$p_2$] {};
	    \node [place] (F) [below right = 1cm and 3cm of p1,label=below:$\Fin$] {};
	    
	    \node [transition, label=below:$t_1$] (t1) at ($(I) + (2,0)$) {};
	    \node [transition, label=below:$t_2$] (t2) at ($(F) + (-2,0)$) {} ;
		\node [transition, label=below:$t_3$,blue] (t3) at ($(p1 -| t1)$) {} ;
	    \node [transition, label=below:$t_4$,blue] (t4) at ($(p1 -| t2)$) {} ;
	    
	    \path[->]
	    (I) edge[] (t1)
	    (t1) edge[] (p1)
	    (t1) edge[] (p2)
	    (t2) edge[] (F)
	    (p1) edge[] node[below] {$2$} (t2)
	    (p2) edge[] node[above] {$2$} (t2)
	    (t3) edge[blue] (p1)
	    (p1) edge[blue] (t4)
	    ;
\end{tikzpicture}
	\caption{An example Petri net with four places $\I$, $p_1$, $p_2$, $\Fin$ and four transitions $t_1$, $t_2$, $t_3$, $t_4$. 
	Arcs pointing to transitions consume tokens from the respective places, and arcs pointing away from transitions produce tokens in the respective places.
	Arcs without labels denote single token consumption or production.
	Other labels, such as `$2$' in this example, are explicit. 
	Initially, the marking can be represented by the vector $(2,0,0,0)$ where there are two tokens in $\I$, and no tokens in the other three places. 
	At this marking, the transition $t_2$ cannot be fired as it needs to consume $1$~token from each of $p_1$ and $p_2$. 
	One can see that by firing a sequence of transitions $(t_1, t_1, t_2)$ we reach the marking $(0,0,0,1)$. 
	The transitions $t_3$ and $t_4$ are highlighted (in blue) because $t_3$ does not consume any tokens and $t_4$ does not produce any tokens.}
	\label{fig:petri-net}
\end{figure}
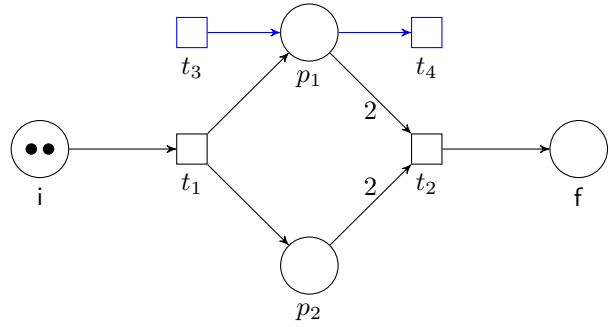

The central decision problem for Petri nets is the \emph{reachability problem}. 
Given a Petri net, an initial marking, and a target marking, reachability asks whether there is a run between the two markings.
Reachability in Petri nets is a decision problem with non-primitive recursive complexity~\cite{Leroux21,LerouxS19}, recently shown to be \class{Ackermann}-complete~\cite{CzerwinskiLLLM21,CzerwinskiO21}. 
The \emph{coverability problem}, a relaxation of the reachability problem, asks whether there is a run that reaches a marking at least as great as the 
target marking.
Provably simpler than reachability; coverability is known to be \class{EXPSPACE}-complete~\cite{Lipton76,DBLP:journals/tcs/Rackoff78,KunnemannMSSW23}. 
In the example (\cref{fig:petri-net}), from the initial marking $(2,0,0,0)$, one can reach $(0,0,0,1)$, but one cannot reach $(0,0,1,0)$. 
However, $(0,0,1,0)$ can be covered since $(1,1,1,0)$ can be reached.

In this paper, we consider Petri nets that are equipped with \emph{resets} but are restricted to be \emph{acyclic}.
Resets are an extra feature of transitions that allow transitions to empty a subset of places.
In modelling processes, resets offer the ability to express cancellation, which is important in many applications~\cite[Table~1]{AalstHHSVVW09}. Unfortunately, in general without any acyclicity restriction, for Petri nets with resets, reachability is undecidable~\cite{ArakiK76, DufourdFS98} and coverability is \class{Ackermann}-complete~\cite{Schnoebelen10,FigueiraFSS11}.
Therefore, in order to observe the decidability of the reachability problem one needs to focus on a subclass of Petri nets with resets.
A~natural restriction is \emph{acyclicity} that applies to the graph representation of the Petri net. 
For example, observe that the Petri net in \cref{fig:petri-net} is acyclic since the arcs do not induce any cycles between the places and transitions.
Both reachability and coverability in acyclic Petri nets are \class{NP}-complete~\cite{LiSGGZ11}.
The \class{NP} upper bound is straightforward: it suffices to guess how many times each transition is fired in the run.
It is always possible to transform this guess into an actual run by sorting the transitions in a topological order induced by the acyclic structure. 
As far as we know, acyclic Petri nets with resets have not been studied previously; they are a natural candidate for an expressive yet tractable class of Petri nets. 
We remark that the \class{NP} upper bound argument for reachability does not translate to the model with resets: changing the order of the resets does not preserve the reached marking.

We study Petri nets and their popular subclass \emph{workflow nets}~\cite{Aalst98}. 
Workflow nets are Petri nets that have two special places, an input place $\I$ and an output place $\Fin$. 
The places and transitions are also restricted so that no tokens can be produced in $\I$, no tokens can be consumed from $\Fin$, and all places and transitions lie on paths from $\I$ to $\Fin$. 
The Petri net in \cref{fig:petri-net} without the (blue) highlighted transitions, $t_3$ and $t_4$, is a workflow net.
Many practical instances of Petri nets are workflow nets~\cite{FahlandFJKLVW09}, which forbid unnatural behaviour.
Workflow nets are well studied~\cite{AalstHHSVVW11}, also with resets~\cite{AalstHHSVVW09}.
The complexities of the reachability and coverability problems for workflow nets are the same as for Petri nets. 
Indeed, the special places $\I$ and $\Fin$ produce and consume the initial and target markings, respectively. 
By introducing additional `artificial' places, it is not challenging to ensure that all places and transitions are on some path from $\I$ to $\Fin$.
However, the last construction does not preserve acyclicity.
It turns out that acyclic Petri nets are more involved than acyclic workflow nets. 
For example, while the set of markings reachable from the initial marking is always finite for acyclic workflow nets~\cite{TipleaBC15}, this is not true for acyclic Petri nets. 
In \cref{fig:petri-net}, place $p_1$ can contain arbitrarily many tokens by firing $t_3$. 
In contrast, the workflow net obtained by removing transitions $t_3$ and $t_4$ will never contain more than $2$ tokens in any place.

\subparagraph{Our results.}
We determine the complexity of reachability and coverability in both \rapn{s} and \rawn{s}.
We prove that coverability in \rapn{s} is \class{PSPACE}-complete.
Further, we show that both reachability and coverability in \rawn{s} are also \class{PSPACE}-complete.
On the other hand, we prove that, rather surprisingly, reachability in \rapn{s} is undecidable.
A summary of our results is in~\cref{fig:complexity-table}.

\begin{figure}
	\centering
	\begin{tabular}{c|c|c}
		& Coverability & Reachability \\ \hline
		\begin{tabular}{c}
			Acyclic workflow \\ nets with resets
		\end{tabular}
		&
		\begin{tabular}{c}
			\class{PSPACE}-complete \\ (\cref{sec:workflow-coverability})
		\end{tabular}
		&
		\begin{tabular}{c}
			\class{PSPACE}-complete \\ (\cref{sec:workflow-reachability})
		\end{tabular}
		\\ \hline
		\begin{tabular}{c}
			Acyclic Petri nets \\ with resets
		\end{tabular}
		&
		\begin{tabular}{c}
			\class{PSPACE}-complete \\ (\cref{sec:petri-coverability})
		\end{tabular}
		&
		\begin{tabular}{c}
		 	Undecidable \\ (\cref{sec:petri-reachability})
		 \end{tabular}
	\end{tabular}
	\caption{A summary of our results.
	\cref{sec:workflow-coverability} contains the \class{PSPACE} lower bound and~\cref{sec:workflow-reachability} and~\cref{sec:petri-coverability} contain the \class{PSPACE} upper bounds.}
	\label{fig:complexity-table}
\end{figure}

For reachability in \rawn{s}, we argue that a place cannot contain more than an exponential number of tokens with respect to the size of the input to the problem. 
The proof is comparable to the proof of the \class{NP} upper bound for acyclic Petri nets: one can reorder the firing sequence of transitions according to a topological order.

\begin{restatable*}{theorem}{workflowreach}\label{thm:workflow-reachability}
Reachability in \rawn{s} is in \class{PSPACE}.
\end{restatable*}

For coverability in \rapn{s}, we show that there are two cases for the number of tokens that a place may contain.
A place may either take at most an exponential number $M$ of tokens, or it can take an arbitrarily large number of tokens, represented by $\omega$.
By abstracting the space of markings to a subset of $\set{ 0, 1, \ldots, M-1, M, \omega}^n$, we can search for a coverability run in polynomial space.

\begin{restatable*}{theorem}{petricover}\label{thm:petri-coverability}
	Coverability in \rapn{s} is in \class{PSPACE}.
\end{restatable*}

We complement these upper bounds with matching lower bounds.
We show that coverability in \rawn{s} requires polynomial space via a polynomial time reduction from QSAT.
In the reduction, we construct an \rawn{} that simulates assignments to the quantified variables (using places whose non-emptiness corresponds to the satisfaction of a literal) and checks that the formula evaluates to true for each assignment (using places whose non-emptiness corresponds to the satisfaction of a clause).

\begin{restatable*}{theorem}{workflowcover}\label{thm:workflow-coverability}
	Coverability in \rawn{s} is \class{PSPACE}-hard.
\end{restatable*}

These three results allow us to conclude that coverability in \rapn{s} and both coverability and reachability in \rawn{s} are \class{PSPACE}-complete problems.
We contrast this with the undecidability of reachability in \rapn{}.
Our proof is a reduction from reachability in general Petri nets with resets, which is known to be undecidable~\cite{ArakiK76}.
The core of our proof is simulating transitions whose arcs are not acyclic.

\begin{restatable*}{theorem}{petrireach}\label{thm:petri-reachability}
	Reachability in \rapn{s} is undecidable.
\end{restatable*}

\subparagraph*{Related work.}
For workflow nets, a central decision problem is the \emph{soundness} problem. 
An instance of soundness usually fixes the initial and target markings to only have one token in $\I$ and $\Fin$, respectively. 
The soundness problem asks whether every marking reachable from the initial marking can then go on to reach the target marking. 
For workflow nets, it is known that soundness reduces to reachability~\cite{Aalst98}, and an optimal algorithm for soundness (which 
does not rely on reachability) was only recently presented~\cite{BlondinMO22LICS}. 
Variants of reachability and coverability have also been used as relaxations to implement soundness~\cite{TipleaM05,BlondinMO22CAV}. 
Thus, we expect this work to provide an initial background to study soundness in \rawn{s} in the future.

In order to obtain decidability for the reachability problem on Petri nets with resets we both restrict the class to workflow nets and enforce acyclicity. 
However, instead of relaxing the class of Petri nets, one could allow the places to contain a negative number of tokens. 
Reachability in this relaxed model is called \emph{integer reachability} and is known to be in \class{NP} for Petri nets\footnote{Integer reachability is \class{NP}-complete for vector addition systems with states~\cite{HaaseKOW09}.}, even with resets~\cite{ChistikovHH18}.

There are many extensions of Petri nets other than adding resets.
We would like to highlight one extension in particular: Petri nets with transfers.
Similar to resets, transfers move all the tokens from one place to another (instead of just removing them)~\cite{ArakiK76}. 
Transfers allow the modelling of some properties of C programs~\cite{0001KW14}.
For Petri nets with transfers, reachability in undecidable~\cite{ArakiK76}, but coverability is decidable~\cite{DufourdFS98}.
More generally, Petri nets with transfers and Petri nets with resets are examples of affine Petri nets~\cite{Valk78,FinkelMP04}. 
The previously mentioned integer reachability problem has been studied for this broad class of Petri nets~\cite{BlondinHMR21,BlondinR21}. 
Consequently, integer reachability in Petri nets with transfers is in \class{PSPACE}~\cite{BlondinHMR21}.
As far as we know, reachability and coverability have not been considered for acyclic Petri nets with transfers or acyclic affine Petri nets,
which we leave as possible future work.

\section{Preliminaries}
\label{sec:preliminaries}

Let $\Z$ be the set of \emph{integers} and $\N$ the set of \emph{natural numbers} (nonnegative integers).
Let $\omega$ stand for the first infinite cardinal, \ie $\omega = \size{\N}$.
Symbols $\Zw$ and $\Nw$ denote the set of natural numbers and the set of integer numbers, each extended with $\omega$, respectively.
As usual, $\size{S}$ denotes the number of elements of a set $S$.
We denote intervals by $[x, y] = \set{z \in \Z \mid x\le z \le y}$.

We use boldface to denote vectors, and we specify a vector by listing its coordinates, which are indexed using square brackets, in a tuple, so $\vec{v} = (\vec v[1], \ldots, \vec v[k])$.
For two vectors $\vec{v}$ and $\vec{w}$ of equal dimension, we write $\vec{v} \geq \vec{w}$ if for every coordinate $s$ we have $\vec{v}[s] \geq \vec{w}[s]$.
If $\vec{v} \geq \vec{w}$ and $\vec{v} \neq \vec{w}$, then $\vec{v} > \vec{w}$; this partial order is called the pointwise order of vectors.
A~vector $\vec{v}$ is non-negative if $\vec{v} \geq (0, 0, \ldots, 0)$.
The norm $\norm{\cdot}$ of a $k$-dimensional vector $\vec{v}$ is the sum of absolute values of its coordinates that are not equal to $\omega$: $\norm{\vec{v}} = \sum_{\vec{v}[i] \in \N} \abs{\vec{v}[i]}$.
We overload notation by saying that the norm $\norm{\cdot}$ of a collection of vectors $V$ is the sum of the norms of vectors in $V$, $\norm{V} = \sum_{\vec{v} \in V} \norm{\vec{v}}$.

\subparagraph*{Petri nets.}
A \emph{Petri net} is a tuple $(P, T, F)$ consisting of a finite set of \emph{places} $P$, a finite set of \emph{transitions} $T$ (disjoint from $P$), and a function defining the \emph{arcs} $F \colon (P \times T) \cup (T \times P) \rightarrow \N$. 
There is an arc from $x$ to $y$ for $(x, y) \in (P \times T) \cup (T \times P)$ if and only if $F(x, y) > 0$. 
In diagrams, this arc is \emph{labelled} with the value $F(x, y)$. 
One can view Petri nets as labelled graphs where $P \cup T$ is the set of nodes, and arcs are edges, labelled according to $F$.
For example, in~\cref{fig:petri-net} for transition $t_1$ we have $F(\I,t_1) = F(t_1,p_1)  = F(t_1,p_2) = 1$ and all other values involving $t_1$ are $0$.
We can define a \emph{path} in a Petri net as a sequence of places and transitions connected by arcs.
A Petri net is \emph{acyclic} if the graph of places and transitions with arcs is acyclic.
The \emph{norm} of a Petri net $\Nn = (P, T, F)$ is $\norm{\Nn} = \abs{P} \cdot \abs{T} + \sum_{p \in P, t \in T} (F(t, p) + F(p, t))$.

\begin{definition}
\label{def:pnr}
    A \emph{Petri net with resets} is a tuple $(P, T, F, R)$, where $(P, T, F)$ is a Petri net and
    $R \colon T \rightarrow 2^P$ is a function defining \emph{reset edges}.
    There is reset edge between a transition $t \in T$ and a place $p \in P$ if and only if $p \in R(t)$.
    A Petri net with resets $(P, T, F, R)$ is an \emph{\rapn{}} if $(P, T, F)$ is acyclic according to the definition above.
\end{definition} 

Importantly, reset edges are \emph{not} subject to the acyclicity restriction. We discuss this in more detail below, after the formal definition of the semantics.

The norm of a Petri net with resets is $\norm{(P, T, F, R)} = \norm{(P, T, F)} + \sum_{t \in T} \abs{R(t)}$.

For a Petri net with resets $(P, T, F, R)$, the \emph{pre-vector} of a transition $t$ is $\pre{t} \colon P \rightarrow \N$, where $\pre{t}[p] = F(p, t)$, and its \emph{post-vector} is $\post{t} \colon P \rightarrow \N$, where $\post{t}[q] = F(t, q)$.
We use similar notation for the \emph{reset-operator} $\reset{t} \subseteq P$, namely $\reset{t} = R(t)$. 

Let us define the semantics of Petri nets (with resets).
The collection of \emph{markings} of a Petri net with resets $(P, T, F, R)$ is the set of all vectors in $\N^P$. 
Places are said to contain \emph{tokens}, a finite resource that can be \emph{consumed}, \emph{produced}, and \emph{reset} by transitions.
For a given marking $\vec{m}$, a place $p$ contains tokens if $\vec{m}[p] > 0$, otherwise it is \emph{empty}.
A transition $t$ can be \emph{fired} at a marking $\vec{m}$ if and only if $\vec{m} \geq \pre{t}$. 
The firing proceeds through the following phases (see \cref{fig:semantics} for an example):
\begin{itemize}
\item first, tokens are \emph{consumed}, which results in $\vec{m}' = \vec{m} - \pre{t}$; 
\item then, places are \emph{reset}, which results in $\vec{m}''$ where $\vec{m}''[p] = 0$ for all $p \in \reset{t}$ and $\vec{m}''[p] = \vec{m}'[p]$ for all $p \not\in \reset{t}$; 
\item finally, tokens are \emph{produced}, which results in the new marking $\vec{n} = \vec{m}'' + \post{t}$.
\end{itemize}
We write $\run{\vec{m}}{t}{\vec{n}}$.

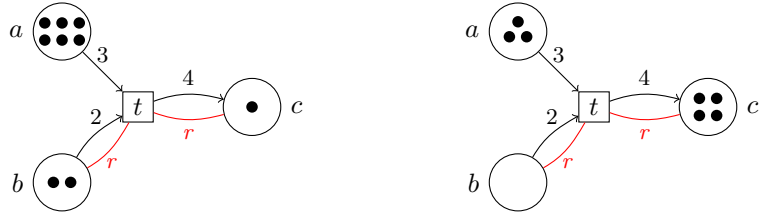
\begin{figure}[t]
    \centering
    \begin{tikzpicture}
	\node[place, tokens = 6, label=left:{$a$}] (p1) at (0, 1) {};
	\node[place, tokens = 2, label=left:{$b$}] (q1) at (0, -1) {};
	\node[transition, inner sep = 2pt] (t1) at (1, 0) {$t$};
	\node[place, tokens = 1, label=right:{$c$}] (r1) at (2.5, 0) {};

	\path[->] (p1) edge[] node[above] {\footnotesize 3} (t1);
	\path[->] (q1) edge[out = 60, in = 210] node[above] {\footnotesize 2} (t1);
	\path[-] (t1) edge[out = 240, in = 30, red] node[below] {\footnotesize \textcolor{red}{\emph{r}}} (q1);
	\path[->] (t1) edge[out = 20, in = 165] node[above] {\footnotesize 4} (r1);
	\path[-] (t1) edge[out = 340, in = 195, red] node[below] {\footnotesize \textcolor{red}{\emph{r}}} (r1);

	\node[place, tokens = 3, label=left:{$a$}] (p2) at (6, 1) {};
	\node[place, label=left:{$b$}] (q2) at (6, -1) {};
	\node[transition, inner sep = 2pt] (t2) at (7, 0) {$t$};
	\node[place, tokens = 4, label=right:{$c$}] (r2) at (8.5, 0) {};

	\path[->] (p2) edge[] node[above] {\footnotesize 3} (t2);
	\path[->] (q2) edge[out = 60, in = 210] node[above] {\footnotesize 2} (t2);
	\path[-] (t2) edge[out = 240, in = 30, red] node[below] {\footnotesize \textcolor{red}{\emph{r}}} (q2);
	\path[->] (t2) edge[out = 20, in = 165] node[above] {\footnotesize 4} (r2);
	\path[-] (t2) edge[out = 340, in = 195, red] node[below] {\footnotesize \textcolor{red}{\emph{r}}} (r2);
\end{tikzpicture}
    \caption{
        Two markings of an \rapn{} with three places $a$, $b$, and $c$.
        Left: upon firing~$t$ from marking $(6,2,1)$ shown, 3 tokens are consumed from $a$ and 2 tokens are consumed from $b$.
        Then, $b$ and $c$ are reset to 0 tokens, from 0 tokens and 1 token, respectively.
        Finally, 4 tokens are produced in $c$; this is the only number of tokens $c$ can contain after $t$ is fired.
        Right: marking $(3, 0, 4)$ is reached as the result of firing~$t$.
    }
    \label{fig:semantics}
\end{figure}

\begin{note*}
    In the semantics of Petri nets with resets, whether or not a place under reset contains tokens does not affect whether the transition can be fired.
    This makes the effect of resets distinct from the usual consumption of tokens by a transition.
    Resets do not produce any tokens either.
    Thus, resets are considered `undirected', and we refer to reset \emph{edges} (rather than arcs).
    For the sake of clarity, in all drawings of Petri nets with resets, the reset edges are undirected and will be coloured red to distinguish them further.
\end{note*}

A \emph{firing sequence} $\sigma = (t_1, t_2, \ldots, t_n)$ is a sequence of transitions. 
It forms a \emph{run} from a marking $\vec m_0 $ to a marking $\vec m_n $ if
$\vec{m}_0 \overset{t_1}{\rightarrow} \vec{m}_1 \overset{t_2}{\rightarrow} \vec{m}_2 \overset{t_3}{\rightarrow} \cdots \overset{t_n}{\rightarrow} \vec{m}_n$ for some intermediate markings $\vec{m}_1, \ldots, \vec{m}_{n-1}$. 
The run is denoted $\run{\vec{m}_0}{\sigma}{\vec{m}_n}$.
We also write $\run{\vec{m}}{*}{\vec{n}}$ if there exists a run from $\vec{m}$ to $\vec{n}$; in this case we say that $\vec{n}$ is \emph{reachable} from $\vec{m}$. 
Further, we say that a run $\run{\vec{m}}{*}{\vec{n}'}$ \emph{covers} $\vec{n}$ if $\vec{n}' \geq \vec{n}$. 
If such a $\sigma$ exists, we say that $\vec{n}$ can be \emph{covered} from~$\vec{m}$.

\begin{note*}
    Every Petri net can be seen as a Petri net with resets whose reset function is null, $R(t) = \emptyset$ for all $t$.
    So all definitions for Petri nets with resets naturally extend to Petri nets.
\end{note*}

\subparagraph*{Workflow nets.}
A \emph{workflow net} is a triple $(\Pp, \I, \Fin)$ where $\Pp$ is a Petri net $(P, T, F)$, $\I \in P$ is the \emph{initial place}, $\Fin \in P$ is the \emph{final place}, and all places and transitions lie on paths from $\I$ to~$\Fin$.
A~\emph{workflow net with resets} is a triple $(\Rr, \I, \Fin)$ where $\Rr = (P, T, F, R)$ is a Petri net with resets, and $((P, T, F), \I, \Fin)$ is a workflow net. 
We say that a workflow net (with resets) $(\Nn, \I, \Fin)$ is \emph{acyclic} if the Petri net (with resets) $\Nn$ is acyclic. In~\cref{fig:petri-net} the Petri net without transitions $t_3$ and $t_4$ is also a workflow net.

\subparagraph*{Decision problems.}
The following problems can be posed with any combination of added resets, acyclicity, and the workflow restriction.

\vspace{0.2cm}
\problemx{Reachability in Petri nets}
{A Petri net $\Nn$, an initial marking $\vec{m}$, and a target marking $\vec{n}$.}
{Does there exist a firing sequence $\sigma$ such that $\run{\vec{m}}{\sigma}{\vec{n}}$?}

\vspace{0.2cm}
\problemx{Coverability in Petri nets}
{A Petri net $\Nn$, an initial marking $\vec{m}$, and a target marking $\vec{n}$.}
{Does there exist a firing sequence $\sigma$ such that $\run{\vec{m}}{\sigma}{\vec{n'}}$, where $\vec{n'} \geq \vec{n}$?}

\vspace{0.2cm}
To give \emph{instances} of these problems, we use tuples $(\Nn, \vec{m}, \vec{n})$. 
The \emph{norm} of an instance is $\norm{(\Nn, \vec{m}, \vec{n})} = \norm{\Nn} + \norm{\vec{m}} + \norm{\vec{n}}$.
Depending on whether the arc weights are written in unary or binary, the \emph{bit size} of the input is polynomial in the norm or logarithmic in the norm, respectively.
Unary encoding suffices for our \class{PSPACE} lower bound (\cref{thm:workflow-coverability}); see~\cref{lem:binary-transitions}. 
Both of our \class{PSPACE} upper bounds (\cref{thm:workflow-reachability} and \cref{thm:petri-coverability}) hold even when the arc weights are binary-encoded.
The undecidability result (\cref{thm:petri-reachability}) is independent of the encoding.

\section{Upper Bounds}
\label{sec:upper-bounds}

\subsection{Reachability in \RAWN{s}}
\label{sec:workflow-reachability}

\workflowreach

\begin{proof}
	We rely on the simple property that reachable markings in \rawn{s} are exponentially bounded.
	Let $\Rr = (P, T, F, R)$ be a given \rawn{} and fix an initial marking $\m$. 
	Consider the workflow net $\Ww = (P, T, F)$ that is just $\Rr$ with the resets removed. 
	Suppose from a marking $\vec{p}$ in $\Rr$, firing a transition $t$ leads to marking $\vec{q}$.
	Clearly with the resets removed, firing $t$ from $\vec{p}$ in $\Ww$ leads to a marking $\vec{q}'$ and $\vec{q}' \geq \vec{q}$.
	Notice also that the removal of resets does not alter whether or not a transition can be fired, if a transition can be fired from $\vec{p}$ then it can be fired from any $\vec{p}' \geq \vec{p}$.
	It follows that if $\run{\vec{m}}{\pi}{\vec{n}}$ in $\Rr$, then $\run{\vec{m}}{\pi}{\vec{n}'}$ in $\Ww$ for some $\vec{n}' \geq \vec{n}$
	With this in mind, it suffices to argue that any reachable marking in $\Ww$ can be stored in polynomial space, relative to the norms of $\vec{m}$ and $\Ww$.

	Let $m = \norm{\Rr} + \norm{\vec{m}}$.
	We prove that if $\run{\vec{m}}{\pi}{\vec{n}'}$ in $\Ww$, then $\norm{\vec{n}'} \leq m^{n+1}$, where $n$ is the number of distinct transitions occurring in the firing sequence $\pi$. 
	Since $\Ww$ is acyclic, there is a topological order on (the sources of) the transitions, and $\pi$ can be permuted to respect this order (cf.~\cite{HiraishiI88}).
	Every transition in a workflow net must consume at least one token, so it follows that the $i$-th distinct transition can be fired at most $m^i$ many times, resulting in a marking of norm at most $m^{i+1}$. 
	Therefore, the norm of the largest possible marking is $m^{n+1}$ and since $n \leq \abs{T}$, all markings observed in the (permuted) run can be written down using polynomially many bits.
	Hence, reachability in \rawn{s} can be decided using polynomial space.
\end{proof}

\subsection{Coverability in \RAPN{s}}
\label{sec:petri-coverability}

\petricover

We fix our attention on an instance $(\Pp, \vec{m}, \vec{n})$ of coverability in \rapn{s}.
Our approach can be summarised in two parts.
First, we construct another infinite-state system $\Nn$ by modifying $\Pp$, that is much like a Petri net. 
The difference is that the places of $\Nn$ may contain an `infinite' number of tokens, denoted $\omega$.
Importantly, we will argue that $\vec{n}$ is coverable from $\vec{m}$ in $\Pp$ if and only if $\vec{n}$ is coverable from $\vec{m}$ in $\Nn$. 
Second, we show that the set of markings reachable from $\vec{m}$ in $\Nn$ has cardinality exponential in $\norm{\Nn}$ and $\norm{\vec{m}}$.
Together, this allows us to decide, in polynomial space, this instance of coverability in \rapn{s}.

We say that a transition $t$ is \emph{generating from a marking $\vec{r}$} if it only consumes tokens from places which contain $\omega$ tokens, more precisely, for each place $p$ such that $\pre{t}[p] > 0$, we have $\vec{r}[p] = \omega$.
In other words, a generating transition can only decrease the number of tokens in some place by resetting it; notice that consuming a finite number of tokens from a place that contains $\omega$ tokens leaves $\omega$ tokens is that place.
Suppose that $\run{\vec{p}}{t}{ \run{\vec{q}}{t}{\vec{r}} }$, where $t$ is a generating transition in $\vec{p}$, then $\vec{r} \geq \vec{q}$.
Indeed, if some place is reset by $t$ then by immediately firing $t$ again, the number of tokens in such a place does not decrease below zero.
By definition, the number of tokens in places that $t$ only consumes from is $\omega$, both before and after firing a generating transition $t$.
Finally, the number of tokens in all other places can only increase after firing $t$ again.
By firing $t$ an arbitrary number of times, the places where $t$ only produces tokens to will then contain $\omega$ many tokens.

Formally, $\Nn$ is the same object as $\Pp$: it consists of the same sets of places, transitions, and resets, but its semantics differs.
A marking $\vec{m}$ of $\Nn$ is allowed to have $\omega$ tokens in its places, so $\vec{m} \in \Nw^n$, where $n$ is the number of places.
Recall that $\omega$ denotes the first infinite cardinal, so $\omega + z = \omega$ for all $z \in \Z$.
To define the semantics of $\Nn$, we need to specify the behaviour of its transitions.
Fix a marking $\vec{m}$.
As is the case in $\Pp$, a transition $t$ can be fired in $\Nn$ if, for every place $p$, $\vec{m}[p] \geq \pre{t}[p]$.
The marking reached depends on whether $t$ is generating from $\vec{m}$.
If $t$ is not generating from $\vec{m}$, then its behaviour is defined as it was in $\Pp$; first subtract $\pre{t}$, then perform the resets, and lastly add $\post{t}$.
Otherwise, if $t$ is generating from $\vec{m}$, then $\run{\vec{m}}{t}{\vec{n}}$ is defined so that
\begin{equation*}
	\vec{n}[p] =
	\begin{cases}
	\omega & \text{if } p \notin \reset{t}, \text{ and either } \post{t}[p] \geq 1 \text{ or } \vec{m}[p] = \omega; \\
	\post{t}[p] & \text{if } p \in \reset{t}; \\
	\vec{m}[p] & \text{otherwise.}
	\end{cases}
\end{equation*}
Intuitively, the transition is applied arbitrarily many times producing $\omega$ tokens to some places, whenever it is possible.
\cref{clm:coverability-holds}, which is proved in~\cref{app:petri-coverability}, allows us to instead decide the coverability instance in $\Nn$ with the abstracted space of configurations.

\begin{restatable}{claim}{coverabilityholds}\label{clm:coverability-holds}
	Let $\vec{m}, \vec{n} \in \N^n$. Then $\vec{n}$ is coverable from $\vec{m}$ in $\Pp$ if and only if $\vec{n}$ is coverable from $\vec{m}$ in $\Nn$.
\end{restatable}

Recall the norm of a marking $\norm{\vec{v}} = \sum_{\vec{v}[p] \in \N} \abs{\vec{v}[p]}$.
\cref{clm:exponential-markings}, which is proved in~\cref{app:petri-coverability}, shows that because $\Nn$ is acyclic, only markings with an exponential norm can be reached.
Note, critically, that places containing $\omega$ tokens do not contribute to the norm.

\begin{restatable}{claim}{exponentialmarkings}\label{clm:exponential-markings}
	Let $k$ be the greatest number of tokens produced by a transition in~$\Pp$ and let $C = \norm{\vec{m}}$.
	If $\run{\vec{m}}{*}{\vec{v}}$ in $\Nn$, then $\norm{\vec{v}} \leq C \cdot k^n$.
\end{restatable}

\begin{proof}[Proof of~\cref{thm:petri-coverability}]
	By~\cref{clm:coverability-holds}, it suffices to show that coverability in the modified \rapn{} $\Nn$ can be decided in polynomial space.
	We can do this by non-deterministically exploring the markings that are reachable from $\vec{m}$ in $\Nn$.
	Given~\cref{clm:exponential-markings}, if $\vec{v}$ is reachable from $\vec{m}$ in $\Nn$, then $\norm{\vec{v}} \leq C \cdot k^n$.
	These reachable markings can be written down using polynomially many bits since $n$ is the number of places, $k$ is the greatest number of tokens produced by a place, and $C$ is the norm of $\vec{m}$.
	Thus, coverability in \rapn{s} is in \class{PSPACE}.
\end{proof}

\section{Lower Bounds}
\label{sec:lower-bounds}

\subsection{Coverability in \RAWN{s}}
\label{sec:workflow-coverability}

With the coverability objective, binary-encoded transitions can be weakly simulated by unary-encoded transitions.
We do this for convenience, since the later reductions can be more succinctly presented with binary-encoded transitions.

\begin{restatable}{lemma}{binarytransitions}\label{lem:binary-transitions}
Given an instance of coverability in \rawn{s} with binary-encoded transitions $I = (\Bb, \vec{m}, \vec{n})$, one can construct, in polynomial time, an instance of coverability in \rawn{s} $I' = (\Uu, \vec{x}, \vec{y})$ with unary-encoded transitions such that $I$ is positive if and only if $I'$ is positive.
\end{restatable}

\workflowcover

All lemmata and claims in this sections are proved in~\cref{app:workflow-coverability}.

\paragraph*{Proof Approach}

We will reduce from the quantified satisfiability (QSAT) problem.
\problemx{QSAT}
{A quantified Boolean formula $\phi$ in conjunctive normal form over $y_1, x_1, \ldots, y_k, x_k$.}
{Does $\forall y_1 \, \exists x_1 \, \ldots \, \forall y_k \, \exists x_k: \phi(y_1, x_1, \ldots, y_k, x_k)$ evaluate to true?}
\vspace{0.1in}

Given a Quantified Boolean Formula (QBF) $\phi$, we will construct an \rawn{} $\Ww$ that mimics the exhaustive approach to verifying $\phi$.
There will be a collection of places that represent an assignment to the variables $y_1, x_1, \ldots, y_k, x_k$.
There will be transitions that consume tokens from these places and produce tokens into a component of $\Ww$ that is used to test whether the current assignment is satisfying.
If the current assignment is satisfying, then one token can be produced to some final place which counts the number of satisfying assigments observed.
The places representing an assignment are controlled by a series of gadgets that we call universal gadgets and existential gadgets. 
In combination, the universal gadgets iterate through each possible assignment to the universal variables and the existential gadgets assign a (nondeterministically chosen) value the existential variables.
A marking in which the final place contains $2^k$ tokens can only be reached if and only if every considered assignment has been checked to be satisfying.
A detailed description of the coverability instance follows.
The proof of correctness consists of two parts.

First, we will verify that the QBF evaluates to true given that coverability holds.
We achieve this via an inductive argument that tracks the simulated assignments to variables over parts of the run witnessing coverability.

In the second part, we would like to recover a firing sequence for coverability if the QBF evaluates to true.
We achieve this by using (partial) assignments to variables in the QBF to inform which transitions need be fired to make progress towards the final marking.

\paragraph*{Construction of the \RAWN{}}

For this section, we focus our attention on a QBF
\begin{equation*}
\forall y_1 \, \exists x_1 \, \forall y_2 \, \exists x_2 \, \ldots \, \forall y_k \, \exists x_k: \phi(y_1, x_1, y_2, x_2, \ldots, y_k, x_k). 
\end{equation*}
We remark that we can add `dummy' clauses $(\overline y_i \lor y_i)$ and $(\overline x_i \lor x_i)$ for each $i \in [1, k]$ to $\phi$ without changing any valuation.

For the proof of~\cref{thm:workflow-coverability}, we construct an \rawn{} $\Ww = (P, T, F, R)$ from the QBF; we first list the places and transitions including resets of $\Ww$.
See~\cref{fig:qbf-net} for an example.

\begin{figure}[ht!]
	\centering
	\begin{tikzpicture}
	\draw[purple, line width = 0.3mm, dotted] (-0.25, -1.8) rectangle (13, -4.7);
	\node at (1.2, -5) {\bf\textcolor{purple}{Clause places $C$}};

	\draw[blue, line width = 0.3mm, dotted] (-0.5, 5) rectangle (1.5, -3);
	\node at (0.5, 5.8) {\bf\textcolor{blue}{Universal}};
	\node at (0.5, 5.3) {\bf\textcolor{blue}{gadget $U_1$}};
	\draw[ForestGreen, line width = 0.3mm, dotted] (1.75, 5) rectangle (3.75, -3);
	\node at (2.75, 5.8) {\bf\textcolor{ForestGreen}{Existential}};
	\node at (2.75, 5.3) {\bf\textcolor{ForestGreen}{gadget $E_1$}};

	\draw[blue, line width = 0.3mm, dotted] (4.25, 5) rectangle (6.25, -3);
	\node at (5.25, 5.8) {\bf\textcolor{blue}{Universal}};
	\node at (5.25, 5.3) {\bf\textcolor{blue}{gadget $U_2$}};
	\draw[ForestGreen, line width = 0.3mm, dotted] (6.5, 5) rectangle (8.5, -3);
	\node at (7.5, 5.8) {\bf\textcolor{ForestGreen}{Existential}};
	\node at (7.5, 5.3) {\bf\textcolor{ForestGreen}{gadget $E_2$}};

	\draw[blue, line width = 0.3mm, dotted] (9, 5) rectangle (11, -3);
	\node at (10, 5.8) {\bf\textcolor{blue}{Universal}};
	\node at (10, 5.3) {\bf\textcolor{blue}{gadget $U_3$}};
	\draw[ForestGreen, line width = 0.3mm, dotted] (11.25, 5) rectangle (13.25, -3);
	\node at (12.25, 5.8) {\bf\textcolor{ForestGreen}{Existential}};
	\node at (12.25, 5.3) {\bf\textcolor{ForestGreen}{gadget $E_3$}};

	\node[place] (h1) at (0, 4.5) {$h_1$};
	\node[place] (h2) at (4.75, 4.5) {$h_2$};
	\node[place] (h3) at (9.5, 4.5) {$h_3$};

	\node[place] (w1) at (1, 2.5) {$w_1$};
	\node[place] (w2) at (5.75, 2.5) {$w_2$};
	\node[place] (w3) at (10.5, 2.5) {$w_3$};

	\node[place] (v1) at (2.75, 3.5) {$v_1$};
	\node[place] (v2) at (7.5, 3.5) {$v_2$};
	\node[place] (v3) at (12.25, 3.5) {$v_3$};

	\node[transition, inner sep = 2pt] (u1b) at (0, 3.5) {\small $u_1^\bot$};
	\node[transition, inner sep = 2pt] (u2b) at (4.75, 3.5) {\small $u_2^\bot$};
	\node[transition, inner sep = 2pt] (u3b) at (9.5, 3.5) {\small $u_3^\bot$};

	\node[transition, inner sep = 2pt] (u1t) at (1, 1.5) {\small $u_1^\top$};
	\node[transition, inner sep = 2pt] (u2t) at (5.75, 1.5) {\small $u_2^\top$};
	\node[transition, inner sep = 2pt] (u3t) at (10.5, 1.5) {\small $u_3^\top$};

	\node[transition, inner sep = 2pt] (e1b) at (2.25, 1.5) {\small $e_1^\bot$};
	\node[transition, inner sep = 2pt] (e1t) at (3.25, 1.5) {\small $e_1^\top$};
	\node[transition, inner sep = 2pt] (e2b) at (7, 1.5) {\small $e_2^\bot$};
	\node[transition, inner sep = 2pt] (e2t) at (8, 1.5) {\small $e_2^\top$};
	\node[transition, inner sep = 2pt] (e3b) at (11.75, 1.5) {\small $e_3^\bot$};
	\node[transition, inner sep = 2pt] (e3t) at (12.75, 1.5) {\small $e_3^\top$};

	\node[place] (b1n) at (0, 0.25) {$\overline b_1$};
	\node[place] (b1p) at (1, 0.25) {$b_1$};
	\node[place] (a1n) at (2.25, 0.25) {$\overline a_1$};
	\node[place] (a1p) at (3.25, 0.25) {$a_1$};

	\node[place] (b2n) at (4.75, 0.25) {$\overline b_2$};
	\node[place] (b2p) at (5.75, 0.25) {$b_2$};
	\node[place] (a2n) at (7, 0.25) {$\overline a_2$};
	\node[place] (a2p) at (8, 0.25) {$a_2$};

	\node[place] (b3n) at (9.5, 0.25) {$\overline b_3$};
	\node[place] (b3p) at (10.5, 0.25) {$b_3$};
	\node[place] (a3n) at (11.75, 0.25) {$\overline a_3$};
	\node[place] (a3p) at (12.75, 0.25) {$a_3$};

	\path[->] (h1) edge[] (u1b);
	\path[->] (h2) edge[] (u2b);
	\path[->] (h3) edge[] (u3b);

	\path[->] (w1) edge[] (u1t);
	\path[->] (w2) edge[] (u2t);
	\path[->] (w3) edge[] (u3t);

	\path[->] (u1b) edge[] (w1);
	\path[->] (u2b) edge[] (w2);
	\path[->] (u3b) edge[] (w3);

	\path[->] (u1b) edge[] (v1);
	\path[->] (u2b) edge[] (v2);
	\path[->] (u3b) edge[] (v3);
	\path[->] (u1t) edge[] (v1);
	\path[->] (u2t) edge[] (v2);
	\path[->] (u3t) edge[] (v3);

	\path[->] (u1b) edge[] node[below right] {\small$2^2$} (b1n);
	\path[->] (u2b) edge[] node[below right] {\small$2^1$} (b2n);
	\path[->] (u3b) edge[] node[below right] {\small$2^0$} (b3n);

	\path[->] (u1t) edge[] node[right] {\small$2^2$} (b1p);
	\path[->] (u2t) edge[] node[right] {\small$2^1$} (b2p);
	\path[->] (u3t) edge[] node[right] {\small$2^0$} (b3p);

	\path[->] (v1) edge[] (e1b);
	\path[->] (v1) edge[] (e1t);
	\path[->] (v2) edge[] (e2b);
	\path[->] (v2) edge[] (e2t);
	\path[->] (v3) edge[] (e3b);
	\path[->] (v3) edge[] (e3t);

	\path[->] (e1b) edge[] node[right] {\small$2^2$} (a1n);
	\path[->] (e1t) edge[] node[right] {\small$2^2$} (a1p);
	\path[->] (e2b) edge[] node[right] {\small$2^1$} (a2n);
	\path[->] (e2t) edge[] node[right] {\small$2^1$} (a2p);
	\path[->] (e3b) edge[] node[right] {\small$2^0$} (a3n);
	\path[->] (e3t) edge[] node[right] {\small$2^0$} (a3p);

	\path[->] (e1b) edge[] (h2);
	\path[->] (e1t) edge[] (h2);
	\path[->] (e2b) edge[] (h3);
	\path[->] (e2t) edge[] (h3);

	\node[transition, inner sep = 2pt] (l1n) at (0, -1) {$\ell_{\overline y_1}$};
	\node[transition, inner sep = 2pt] (l1p) at (1, -1) {$\ell_{y_1}$};
	\node[transition, inner sep = 2pt] (m1n) at (2.25, -1) {$\ell_{\overline x_1}$};
	\node[transition, inner sep = 2pt] (m1p) at (3.25, -1) {$\ell_{x_1}$};

	\node[transition, inner sep = 2pt] (l2n) at (4.75, -1) {$\ell_{\overline y_2}$};
	\node[transition, inner sep = 2pt] (l2p) at (5.75, -1) {$\ell_{y_2}$};
	\node[transition, inner sep = 2pt] (m2n) at (7, -1) {$\ell_{\overline x_2}$};
	\node[transition, inner sep = 2pt] (m2p) at (8, -1) {$\ell_{x_2}$};

	\node[transition, inner sep = 2pt] (l3n) at (9.5, -1) {$\ell_{\overline y_3}$};
	\node[transition, inner sep = 2pt] (l3p) at (10.5, -1) {$\ell_{y_3}$};
	\node[transition, inner sep = 2pt] (m3n) at (11.75, -1) {$\ell_{\overline x_3}$};
	\node[transition, inner sep = 2pt] (m3p) at (12.75, -1) {$\ell_{x_3}$};

	\path[->] (b1n) edge[] (l1n);
	\path[->] (b1p) edge[] (l1p);
	\path[->] (a1n) edge[] (m1n);
	\path[->] (a1p) edge[] (m1p);
	\path[->] (b2n) edge[] (l2n);
	\path[->] (b2p) edge[] (l2p);
	\path[->] (a2n) edge[] (m2n);
	\path[->] (a2p) edge[] (m2p);
	\path[->] (b3n) edge[] (l3n);
	\path[->] (b3p) edge[] (l3p);
	\path[->] (a3n) edge[] (m3n);
	\path[->] (a3p) edge[] (m3p);

	\node[place] (dy1) at (0.5, -2.5) {$d_{y_1}$};
	\node[place] (dx1) at (2.75, -2.5) {$d_{x_1}$};
	\node[place] (dy2) at (5.25, -2.5) {$d_{y_2}$};
	\node[place] (dx2) at (7.5, -2.5) {$d_{x_2}$};
	\node[place] (dy3) at (10, -2.5) {$d_{y_3}$};
	\node[place] (dx3) at (12.25, -2.5) {$d_{x_3}$};

	\path[->] (l1n) edge[] (dy1);
	\path[->] (l1p) edge[] (dy1);
	\path[->] (m1n) edge[] (dx1);
	\path[->] (m1p) edge[] (dx1);
	\path[->] (l2n) edge[] (dy2);
	\path[->] (l2p) edge[] (dy2);
	\path[->] (m2n) edge[] (dx2);
	\path[->] (m2p) edge[] (dx2);
	\path[->] (l3n) edge[] (dy3);
	\path[->] (l3p) edge[] (dy3);
	\path[->] (m3n) edge[] (dx3);
	\path[->] (m3p) edge[] (dx3);

	\node[place] (c1) at (2.55, -4) {$c_1$};
	\node[place] (c2) at (5.1, -4) {$c_2$};
	\node[place] (c3) at (7.65, -4) {$c_3$};
	\node[place] (c4) at (10.2, -4) {$c_4$};

	\path[->] (l1p) edge[bend right = 9] (c1);
	\path[->] (m1n) edge[bend right = 12] (c1);
	\path[->] (l2n) edge[bend left = 9] (c1);

	\path[->] (m1n) edge[bend left = 7] (c2);
	\path[->] (l2n) edge[bend right = 11] (c2);
	\path[->] (m2p) edge[bend right = 10] (c2);	
	
	\path[->] (l2p) edge[bend right = 8] (c3);
	\path[->] (m2p) edge[bend left = 10] (c3);
	\path[->] (l3n) edge[bend left = 7] (c3);	

	\path[->] (m2p) edge[bend right = 8] (c4);
	\path[->] (l3p) edge[bend left = 10] (c4);
	\path[->] (m3p) edge[bend right = 5] (c4);

	\node[transition, inner sep = 2pt, minimum width = 8mm, minimum height = 5mm] (s) at (6.375, -5.5) {\large $s$};

	\path[->] (c1) edge[bend right = 6] (s);
	\path[->] (c2) edge[bend left  = 5] (s);
	\path[->] (c3) edge[bend right = 5] (s);
	\path[->] (c4) edge[bend left = 6] (s);	
	\path[->] (dy1) edge[bend right = 30] (s);
	\path[->] (dx1) edge[bend right = 8] (s);
	\path[->] (dy2) edge[bend left  = 8] (s);
	\path[->] (dx2) edge[bend right  = 8] (s);
	\path[->] (dy3) edge[bend left = 8] (s);
	\path[->] (dx3) edge[bend left = 30] (s);

	\node[place] (f) at (6.375, -6.5) {$f$};

	\path[->] (s) edge[] (f);
\end{tikzpicture}
	\caption{
	The \rawn{} $\Ww$, drawn without resets for sake of clarity, for the QBF $\forall y_1 \exists x_1 \forall y_2 \exists x_2 \forall y_3 \exists x_3: \phi(y_1, x_1, y_2, x_2, y_3, x_3)$ where $\phi(y_1, x_1, y_2, x_2, y_3, x_3) = (y_1 \vee \overline x_1 \vee \overline y_2) \wedge (\overline x_1 \vee \overline y_2 \vee x_2) \wedge(y_2 \vee x_2 \vee \overline y_3 ) \wedge (x_2 \vee y_3 \vee x_3) \wedge (y_1 \vee \overline y_1) \wedge (x_1 \vee \overline x_1) \wedge (y_2 \vee \overline y_2) \wedge (x_2 \vee \overline x_2) \wedge (y_3 \vee \overline y_3) \wedge (x_3 \vee \overline x_3)$.
	All universal and existential control transitions reset all later occurring places in the universal and existential gadgets and in all clause places.
	The loading transitions reset all later occurring dummy clause places.
	The satisfaction transition resets all clause places.
	} 
	\label{fig:qbf-net}
\end{figure}

\subparagraph*{The places.}
There is a place for each literal: for every $i \in [1, k]$, there is $b_i$ for $y_i$, $\overline b_i$ for $\overline y_i$, $a_i$ for $x_i$, and $\overline a_i$ for $\overline x_i$.
Let $L$ denote the set of the literal places.
The non-emptiness of the place $\overline b_i$, for example, will represent assigning false to the variable $y_i$.

There is a place for each clause: for every $j \in [1, m]$, there is $c_j$ for the $j$-th clause.
Furthermore, for every $i \in [1, k]$, there is $d_{y_i}$ for the dummy clause $(y_i \vee \overline y_i)$ and there is $d_{x_i}$ for the dummy clause $(x_i \vee \overline x_i)$.
All \emph{clause places} $c_1, \ldots, c_m, d_{x_1}, d_{y_1}, \ldots, d_{x_k}, d_{y_k}$ are distinct; the set comprising them is denoted $C$.
The non-emptiness of a clause place $c \in C$ will represent whether the corresponding clause has been satisfied.

For each $i \in [1, k]$, there is a \emph{holding place} $h_i$ and a \emph{waiting place} $w_i$ for each universally quantified variable, as well as a \emph{decision place} $v_i$ for each existentially quantified variable.
If the holding place $h_i$ contains a token, one should think that the universally quantified variable $y_i$ and all subsequent variables $x_i, y_{i+1}, x_{i+1}, \ldots, y_k, x_k$ have not yet been assigned.
The waiting place $w_i$ will contain a token if the universally quantified variable $y_i$ is currently assigned false.
The decision place $v_i$ contains a token after the truth assignment of the prior universally quantified variable $y_i$ has completed, but the existentially quantified variable $x_i$ has not yet received a value.
The literal places, holding places, waiting places, decision places, and dummy clause places are grouped into \emph{gadgets}. 
There are $k$ \emph{universal gadgets} $U_i = \set{ h_i, w_i, b_i, \overline b_i, d_{y_i} }$ and $k$ \emph{existential gadgets} $E_i = \set{ v_i, a_i, \overline a_i, d_{x_i} }$.

Finally, there is a place $f$ which counts the number of assignments that have been verified to satisfy the QBF.

The initial place $\I$ of the workflow is $h_1$ and the final place $\Fin$ of the workflow is $f$.

\subparagraph*{The transitions.}
Here, binary-encoded transitions are used, see~\cref{lem:binary-transitions}.
The resets will be specified later.

Inside the universal gadget $U_i$, there are two \emph{universal control transitions} $u_i^\bot$ and $u_i^\top$. 
Firing $u_i^\bot$ corresponds to setting $y_i$ to false, and firing $u_i^\top$ corresponds to setting $y_i$ to true.
The transition $u_i^\bot$ consumes one token from $h_i$, produces one token to $w_i$, produces one token to $v_i$, and produces $2^{k-i}$ tokens to $\overline b_i$; the transition $u_i^\top$ consumes one token from $w_i$, produces one token to $v_i$, and produces $2^{k-i}$ tokens to $b_i$.

Inside the existential gadget $E_i$, there are two \emph{existential control transitions} $e_i^\bot$ and $e_i^\top$. 
Firing $e_i^\bot$ corresponds to setting $x_i$ to false, and firing $e_i^\top$ corresponds to setting $x_i$ to true.
The transition $e_i^\bot$ consumes one token from $v_i$, produces $2^{k-i}$ tokens to $\overline a_i$, and produces one token to $h_{i+1}$; similarly, the transition $e_i^\top$ consumes one token from $v_i$, produces $2^{k-i}$ tokens to $a_i$, and produces one token to $h_{i+1}$. 

Informally, the $i$-th universal or existential controlling transitions produce $2^{k-i}$ tokens to places $\overline b_i$, $b_i$, $\overline a_i$, and $a_i$ so that their values are `remembered' whilst the inner quantified variables have their assignments exhausted.

Connecting the universal and existential gadgets to the clause places are a series of \emph{loading transitions}.
There is a loading transition for each literal; for each $i \in [1, k]$, there are transitions $\ell_{\overline y_i}$, $\ell_{y_i}$, $\ell_{\overline x_i}$, and $\ell_{x_i}$.
The loading transition $\ell_{y_i}$, for example, consumes a token from the place $b_i$ and produces a token to each clause place corresponding to a clause containing the literal $y_i$, including the dummy clause place $d_{y_i}$.

There is a \emph{satisfaction transition} $s$ that consumes a token from each of the clause places and produces a token into a final place $f$.
Intuitively, $s$ can only be fired when all of the clauses have been satisfied (and $f$ is used to count the number of satisfying assignments).

\subparagraph*{Ordering places and transitions.}
The following linear ordering \emph{earlier than} (denoted $\prec$) on $P \cup T$ shows that $\Ww$ is acyclic:
\begin{multline*}
	h_1, u_1^\bot, w_1, u_1^\top, \overline b_1, b_1, v_1, e_1^\bot, e_1^\top, \overline a_1, a_1, \ldots, h_k, u_k^\bot, w_k, u_k^\top, \overline b_k, b_k, v_k, e_k^\bot, e_k^\top, \overline a_k, a_k, \\ 
	\ell_{\overline y_1}, \ell_{y_1}, \ell_{\overline x_1}, \ell_{x_1}, \ldots, \ell_{\overline y_k}, \ell_{y_k}, \ell_{\overline x_k}, \ell_{x_k}, d_{y_1}, d_{x_1}, \ldots, d_{y_k}, d_{x_k}, c_1, \ldots, c_m, s, f.
\end{multline*}

\subparagraph*{The resets.}
The universal and existential control transitions reset all \emph{later} occurring places in the universal gadgets and existential gadgets and \emph{all} dummy clause places.
This also includes the places corresponding to the literals; for example, $u_i^\bot$ and $u_i^\top$ reset both $\overline b_i$ and $b_i$, so it is always true that either $\overline b_i$ or $b_i$ is empty (or possibly both).

This effectively forces the universal and existential control transitions to be fired in sequence: $u_1^\bot$ or $u_1^\top$, then $e_1^\bot$ or $e_1^\top$, then $u_2^\bot$ or $u_2^\top$, etc., until $e_k^\bot$ or $e_k^\top$ is fired.

The loading transitions reset all later occurring dummy clause places.
For example, $\ell_{y_i}$ resets $d_{x_i}, d_{y_{i+1}}, d_{x_{i+1}}, \ldots, d_{y_k}, d_{x_k}$.
Similarly, this forces the loading transitions to also be fired in sequence: $\ell_{\overline y_1}$ or $\ell_{y_1}$, then $\ell_{\overline x_1}$ or $\ell_{x_1}$, then $\ell_{\overline y_2}$ or $\ell_{y_2}$, until $\ell_{\overline x_k}$ or $\ell_{x_k}$ is fired.
This is due to the fact that all dummy places must be non-empty to fire the satisfaction transition.

Finally, the satisfaction transition resets all clause places. 
It could be the case that a clause contains two true literals under an assignment, so the clause place contains two tokens.
It is necessary to clear such a place.
Note that the final place $f$ cannot be reset.

\subparagraph*{Coverability instance $(\Ww, \vec m, \vec n)$.}
We have just defined the \rawn{} $\Ww$.
The initial marking $\vec{m}$ only has one token in the initial place; $\vec{m}[h_1] = 1$ and, for all $p \in P \setminus \set{h_1}, \vec{m}[p] = 0$.
The target marking $\vec{n}$ only has $2^k$ tokens in the final place; $\vec{n}[f] = 2^k$ and, for all $p \in P \setminus \set{f}, \vec{n}[p] = 0$.

\paragraph*{Part One: Coverability implies QBF is true}

We would like to prove an inductive statement of the following, informally described, kind.
Consider any run from the initial marking that covers the target marking.
Let $\sigma$ be an infix of this run from $\vec{p}$ to $\vec{q}$, and let $i$ be a number in $[0, k]$ such that $2^i$ divides $\vec{p}[f]$ and that $\vec{q}[f] = \vec{p}[f] + 2^i$.
This means that $\sigma$ fires the satisfaction transition, $s$, $2^i$ many times.
Then the following (partial) QBF is true:
\begin{equation*}
	\forall y_{k-i+1} \; \exists x_{k-i+1} \ldots
 	\forall y_{k}     \; \exists x_{k} :
 	\phi(\beta_1, \alpha_1, \ldots, \beta_{k-i}, \alpha_{k-i}, y_{k-i+1}, x_{k-i+1}, \ldots, y_k, x_k).
\end{equation*}
Here $(\beta_1, \alpha_1, \ldots, \beta_{k-i}, \alpha_{k-i}) \in \set{0, 1}^{2(k-i)}$ is determined by $\vec{p}$.

To realise this plan, we need several ingredients.
For the base case of the induction, $i = 0$: $\sigma$ only fires $s$ once.
We will determine $(\beta_1, \alpha_1, \beta_2, \alpha_2, \ldots, \beta_k, \alpha_k) \in \set{0, 1}^{2k}$ based on $\vec{p}$, in particular on which of the places $\overline b_i$ and $b_i$, as well as $\overline a_i$ and $a_i$, are non-empty in $\vec{p}$.
Note that it might not be sufficient to consider only the marking $\vec{p}$ since this could be, for instance, the initial marking $\vec{m}$, which has all places empty, bar $h_1$.
So the ``existential decisions'' that determine $\alpha_1, \alpha_2, \ldots, \alpha_k$ need to be found from a prefix of $\sigma$.

For the inductive step, $i > 0$: the infix $\sigma$ fires the satisfaction transition $2^i$ times.
We will split $\sigma$ in two: $\sigma_0$ and $\sigma_1$.
We will use the inductive hypothesis on both subruns.
For this to work, we will show that the partial assignments
\begin{align*}
	(\beta_1, \alpha_1, \ldots, \beta_{k-i+1}, \alpha_{k-i+1}) &\in \{0,1\}^{2 (k-i+1)}
	\qquad\text{and}\\
	(\beta'_1, \alpha'_1, \ldots, \beta'_{k-i+1}, \alpha'_{k-i+1}) &\in \{0,1\}^{2 (k-i+1)},
\end{align*}
which are determined based on each half of the run, satisfy the constraints
\begin{quote}
$\beta_1 = \beta'_1$, $\alpha_1 = \alpha'_1$, \ldots, $\beta_{k-i} = \beta'_{k-i}$, $\alpha_{k-i} = \alpha'_{k-i}$,\quad
$\beta_{k-i+1} = 0$,\quad and\quad $\beta'_{k-i+1} = 1$.
\end{quote}
Informally speaking, these partial assignments are complementary with respect to the $i$-th innermost universally quantified variable.
Note that the index variable $i$ is reused in a variety of contexts throughout the following claims. 

\subparagraph*{Properties of markings}
\begin{restatable}{claim}{truexorfalse}\label{clm:true-xor-false}
	If $\vec{v}$ is reachable from $\vec{m}$, then for every $i \in [1, k]$, $\vec{v}[\overline b_i]=0$ or $\vec{v}[b_i]=0$, and $\vec{v}[\overline a_i]=0$ or $\vec{v}[a_i]=0$.
\end{restatable}

\begin{restatable}{claim}{collectingtropies}\label{clm:collecting-trophies}
	If $\run{\vec{p}}{t}{\vec{q}}$, then $\vec{q}[f] - \vec{p}[f] \in \set{0,1}$.
\end{restatable}

Let us define, for each $i \in [1, k]$, two functions $g_i, g'_i:\N^P \rightarrow \N$ that map a marking to a natural number.
We will use these functions to define a collection of \emph{good} markings.
\begin{align*}
	g_i(\vec{v}) \coloneqq {}&
	\vec{v}[f] + \vec{v}[\overline b_i] + \vec{v}[b_i] + \vec{v}[d_{y_i}] + \sum_{j = 1}^{i}2^{k-j} \cdot (2\vec{v}[h_{j}] + \vec{v}[w_j] + \vec{v}[v_j]) - 2^{k-i}\cdot\vec{v}[v_i] \\
	g_i'(\vec{v}) \coloneqq {}&
	\vec{v}[f] + \vec{v}[\overline a_i] + \vec{v}[a_i] + \vec{v}[d_{x_i}] + \sum_{j = 1}^{i}2^{k-j} \cdot (2\vec{v}[h_{j}] + \vec{v}[w_j] + \vec{v}[v_j]) 
\end{align*}
\begin{definition}[Good marking]
        \label{def:good-marking}
	A marking $\vec{v}$ is \emph{good} if for each $i \in [1, k]$, $g_i(\vec{v}) = 2^k$ and $g'_i(\vec{v}) = 2^k$.
	A marking is \emph{bad} if it is not good.
\end{definition}

Roughly speaking, a marking is good if no tokens in the universal gadgets $U_i$ and no tokens in the existential gadgets $E_i$ have been lost due to a reset.
We discuss good markings in more detail in \cref{app:workflow-coverability}.

\begin{restatable}{claim}{invariantcantincrease}\label{clm:invariant-cant-increase}
	Suppose $\run{\vec{p}}{t}{\vec{q}}$, then $g_i(\vec{p}) \geq g_i(\vec{q})$ and $g'_i(\vec{p}) \geq g'_i(\vec{q})$ for each $i \in [1, k]$.
\end{restatable}

\begin{restatable}{claim}{goodinvariant}\label{clm:good-invariant}
	Suppose $\run{\vec{p}}{t}{\vec{q}}$, where $\vec{p}$ is reachable from $\vec{m}$.
	If $\vec{q}$ is good, then $\vec{p}$ is good.
\end{restatable}

Given~\cref{clm:good-invariant} and since the target marking $\vec{n}$ is good, only good markings can be observed on a covering run from the initial marking $\vec{m}$.
From this, we know that if a bad marking is ever reached, the target marking cannot be covered.

\begin{restatable}{claim}{reachtarget}\label{clm:reach-target}
	If $\run{\vec{m}}{\pi}{\vec{n}'}$ where $\vec{n}' \geq \vec{n}$, then $\vec{n}' = \vec{n}$.
\end{restatable}

The following claim shows that resetting any non-empty place in any of the universal or existential gadgets results in a bad marking.
Recall $\prec$, the previously defined \emph{earlier than} ordering of places and transitions.
\begin{restatable}{claim}{badtransitions}\label{clm:bad-transitions}
	Suppose $\run{\vec{p}}{t}{\vec{q}}$ where $\vec{p}$ is reachable and $t \in \set{ u_i^\bot, u_i^\top, e_i^\bot, e_i^\top : i \in [1, k]}$.
	If there exists $p \in U_1 \cup E_1 \cup \cdots \cup U_k \cup E_k$ such that $t \prec p$ and $\vec{p}[p] \geq 1$, then $\vec{q}$ is bad.
\end{restatable}

\paragraph*{Extracting Assignments from Markings}
We will now explain the relationship between markings and partial assignments.
For a good marking $\vec{v}$, let $\mathrm{val}(\vec{v})$ be the vector $(\beta_1, \alpha_1, \beta_2, \alpha_2, \ldots, \beta_k, \alpha_k) \in \set{0,1,?}^{2k}$ such that
\begin{equation*}
	\beta_i \coloneqq
	\begin{cases} 
		0 & \vec{v}[\overline b_i] \geq 1 \\ 
		1 & \vec{v}[b_i] \geq 1 \\
		? & \text{otherwise,}
	\end{cases}
	\quad\text{and}\quad
	\alpha_i \coloneqq
	\begin{cases}
		0 & \vec{v}[\overline a_i] \geq 1 \\ 
		1 & \vec{v}[a_i] \geq 1 \\
		? & \text{otherwise.}
	\end{cases}
\end{equation*}
The intention is that, for every $i \in [1, k]$, $\beta_i$ and $\alpha_i$ correspond to the values of the Boolean variables $y_i$ and $x_i$, respectively.
Note that~\cref{clm:true-xor-false} ensures that $\beta_i$ and $\alpha_i$ are well-defined, since, for example, $\overline b_i$ and $b_i$ cannot both be non-empty in a reachable marking.
Notice that not all good markings correspond to fully defined variable assignments, but only those in which all $h_i$ and $v_i$ are empty.
We will see that $\vec p[h_i] = \vec p[v_i] = 0$ implies that either $\overline b_i$ or $b_i$ and either $\overline a_i$ or $a_i$ are non-empty, except for right at the end, for example when the target marking $\vec{n}$ is reached. 
Conversely, if $h_i$ contains a token, then neither $\overline b_i$ nor $b_i$ will contain a token (one can think that the Boolean variable $y_i$ has not yet been assigned its value).
Only after firing $u_i^\bot$ does it first get assigned false (before later being assigned true when $u_i^\top$ is eventually fired).

Recall that $C \subseteq P$ is the collection of clause places.
We say that a marking $\vec{v}$ is \emph{clause-free} if $\vec{v}[c] = 0$ for all $c \in C$.

\begin{restatable}{lemma}{qbfinduction}\label{lem:qbf-induction}
	Fix $i \in [0, k]$ and suppose $\run{\vec{p}}{\sigma}{\vec{q}}$ and the following properties hold:
	\begin{enumerate}[(1)]
		\item $\vec{p}$ is a clause-free marking that is reachable from $\vec{m}$;
		\item $\vec{n}$ is coverable from $\vec{q}$;
		\item\label{div-and-diff} $2^i$ divides $\vec{p}[f]$ and $\vec{q}[f] = \vec{p}[f] + 2^i$;
		\item\label{last-fire} the last transition of $\sigma$ is $s$;
		\item\label{enough-fuel} for all $j \in [1, k-i]$, $\vec{p}[\overline b_i] + \vec{p}[b_i] \geq 2^i$, and $\vec{p}[\overline a_i] + \vec{p}[a_i] \geq 2^i$;
		\item if $i > 0$, then $\vec{p}[h_{k-i+1}] = 1$; and
		\item if $i > 0$, then, for all $p \in U_{k-i+1} \cup E_{k-i+1} \cup \cdots \cup U_k \cup E_k$ except $h_{k-i+1}$, $\vec{p}[p] = 0$.
	\end{enumerate}
	Let $\mathrm{val}(\vec{p}) = (\beta_1, \alpha_1, \ldots, \beta_k, \alpha_k)$.
        Then the following QBF evaluates to true:
	\begin{equation*}
		\forall y_{k-i+1} \, \exists x_{k-i+1} \, \ldots \, \forall y_k \, \exists x_k : \phi(\beta_1, \alpha_1, \ldots, \beta_{k-i}, \alpha_{k-i}, y_{k-i+1}, x_{k-i+1}, \ldots, y_k, x_k).
	\end{equation*}
	Moreover, $\sigma$ does not fire transitions $u_1^\bot, u_1^\top, e_1^\bot, e_1^\top, \ldots, u_{k-i}^\bot, u_{k-i}^\top, e_{k-i}^\bot, e_{k-i}^\top$.
\end{restatable}

\paragraph*{Part Two: QBF is true implies Coverability}
Here we would like to recover a firing sequence for coverability if the QBF evaluates to true.
Depending on the current assignment of the universally quantified variables, $y_1, \ldots y_i$, and the already selected assignments of the existentially quantified variables $x_1, \ldots, x_{i-1}$, one can use the truth of the QBF to determine whether $x_i$ is assigned true or false. 
This choice informs which of the next existentially quantified transitions to fire.

\begin{restatable}{lemma}{coverabilityrun}\label{lem:coverability-run}
	Fix $i \in [0,k]$ and suppose that for some $\beta_1, \alpha_1, \ldots, \beta_{k-i}, \alpha_{k-i} \in \set{0, 1}$, the following QBF evaluates to true:
	\begin{equation*}
		\forall y_{k-i+1} \, \exists x_{k-i+1} \, \ldots \, \forall y_k \, \exists x_k : \phi(\beta_1, \alpha_1, \ldots, \beta_{k-i}, \alpha_{k-i}, y_{k-i+1}, x_{k-i+1}, \ldots, y_k, x_k).
	\end{equation*}
	Let $\vec{p}$ be a marking such that, if $i > 0$, then $\vec{p}[h_{k-i+1}] = 1$, and, for every $j \in [1, k-i]$:
	\begin{enumerate}[(1)]
		\item if $\beta_j = 0$, then $\vec{p}[\overline b_j] \geq 2^i$, otherwise if $\beta_j = 1$, then $\vec{p}[b_j] \geq 2^i$; and
		\item if $\alpha_j = 0$, then $\vec{p}[\overline a_j] \geq 2^i$, otherwise if $\alpha_j = 1$, then $\vec{p}[a_j] \geq 2^i$.
	\end{enumerate}
	Then there exists a firing sequence $\sigma$ such that $\run{\vec{p}}{\sigma}{\vec{q}}$ where $\vec{q}$ is a marking such that $\vec{q}[f] = \vec{p}[f] + 2^i$, and for every $j \in [1, k-i]$:
	\begin{enumerate}[(a)]
		\item $\vec{q}[\overline b_j] + \vec{q}[b_j] = \vec{q}[\overline b_j] + \vec{q}[b_j] - 2^i$; 
		\item $\vec{q}[\overline a_j] + \vec{q}[a_j] = \vec{q}[\overline a_j] + \vec{q}[a_j] - 2^i$; and
		\item $\vec{q}[h_j] = \vec{p}[h_j]$, $\vec{q}[w_j] = \vec{p}[w_j]$, and $\vec{q}[v_j] = \vec{p}[v_j]$.
	\end{enumerate}
\end{restatable}

\paragraph*{Completing the proof}
\begin{proof}[Proof of~\cref{thm:workflow-coverability}]
    The reduction from QSAT is already outlined above.
	Given an instance of QSAT that consists of a QBF $\phi$ over $y_1, x_1, \ldots, y_k, x_k$, there exists an instance of coverability in \rawn{s} $(\Ww, \vec{m}, \vec{n})$ such that $\forall y_1 \, \exists x_1 \, \ldots \, \forall y_k \, \exists x_k: \phi(y_1, x_1, \ldots, y_k, x_k)$ evaluates to true if and only if $\run{\vec{m}}{*}{\vec{n}'}$ in $\Ww$ where $\vec{n}' \geq \vec{n}$.
	The backwards implication is given by~\cref{lem:qbf-induction} with $i = k$, $\vec{p} = \vec{m}$, and $\vec{q} = \vec{n}'$.
	The forwards implication is given by~\cref{lem:coverability-run} with $i = k$.
\end{proof}

\begin{corollary}
	Reachability in \rawn{s} and coverability in both \rapn{s} and \rawn{} are all \class{PSPACE}-complete.
	\label{cor:pspace-complete}
\end{corollary}

\subsection{Reachability in \RAPN{s}}
\label{sec:petri-reachability}

In this section, we will prove that reachability in \rapn{s} is undecidable.
We reduce from reachability in Petri nets with zero tests, a problem that is well-known to be undecidable, following from the undecidability of reachability in counter machines~\cite{Minsky67}.
A \emph{Petri net with zero tests} is a tuple $(P, T, F, Z)$, where $(P, T, F)$ is a Petri net and $Z: T \rightarrow 2^P$ is a function defining the zero-test edges.
A transition $t \in T$ zero-tests a place $p \in P$ if $p \in Z(t)$. Then $t$ can be fired only if $p$ is empty.
As is the case for resets, an \emph{\zapn{}} does not subject zero-test edges to the acyclicity restriction.

\petrireach

The reduction is split into two parts.
\cref{lem:petri-simulate-transitions} shows how \zapn{s} can simulate (not necessarily acyclic) Petri nets with zero tests.
This requires using zero tests, transitions that do not consume tokens, and transitions that do not produce tokens.
Then, in~\cref{lem:petri-simulate-zero-tests}, we show how \rapn{s} can simulate \zapn{s}.
This requires some additional places and relies on the reachability objective to ensure that zero tests are simulated faithfully.
The proof is very similar to the proof that reachability in Petri nets with resets is undecidable.
We follow through with the construction to make it clear that acyclicity is preserved.

\begin{lemma}
	The reachability problem in Petri nets with zero tests is reducible in logarithmic space to the reachability problem in \zapn{s}.
	\label{lem:petri-simulate-transitions}
\end{lemma}
\begin{proof}
	Let $\Pp = (P, T, F, Z)$ be a Petri net with zero tests.
	We will construct an \zapn{} $\Zz = (P', T', F', Z')$.
	For every transition $t \in T$, we will add two additional places $c_t$ and $p_t$ to the set of places. 
	Formally, we define $G = \set{c_t, p_t \mid t \in T}$ and $P' = P \cup G$.
	For every transition $t \in T$, we create three transitions $\sim$, $\con$, and $\pro$, so $T' = \set{\sim, \con, \pro \mid t \in T}$. 
	The intention is that firing $t \in T$ will be simulated by firing $\sim$, $\con$, and $\pro$ successively.
	\cref{fig:transition-controller} illustrates the construction.
	To define the transitions in detail, fix $t \in T$.

	\begin{itemize}
		\item The transition $\sim$ simulates choosing $t$ to be the next transition. 
		Formally, $\post{\sim}[c_t] = \post{\sim}[p_t] = 1$, and $Z'(\sim) = G$. 
		Note that to fire $\sim$ all places in $G$ must be empty, and upon firing $\sim$, a token is placed in $c_t$ and $p_t$. 
		Thus no other transition $s_\textup{sim}$, $\sim$, or $u_\textup{sim}$ can be fired until the tokens in $c_t$ and $p_t$ are consumed.
		\item The transition $\con$ performs the token consumption and zero tests of $t$.
		Formally, $\pre{\con}[c_t] = 1$, $Z'(\con) = Z(t)$, and $\pre{\con}[p] = \pre{t}[p]$ for every $p \in P$.
		The consumption of the token from $c_t$ indicates that the consumptions and zero tests of $t$ have been actioned.
		\item The transition $\pro$ performs the token productions of $t$. 
		Formally, $\pre{\pro}[p_t] = 1$, $Z'(\con)~=~\set{ c_t }$, and $\post{\pro}[p] = \post{t}[p]$, for each $p \in P$.
		The consumption of the token from $p_t$ indicates that the productions of $t$ have been actioned.
		The zero test on $c_t$ forces a firing order that mimics the semantics of firing $t$.
	\end{itemize}

	Indeed, after firing $\sim$ the only transition that can be fired is $\con$ since all other transitions require $c_t$ to be empty or require a place in $G\setminus\set{c_t}$ to be non-empty.
	Then, after firing $\con$ the only transition that can be fired is $\pro$ since all other transitions either require $p_t$ to be empty, or require a place in $G\setminus\set{p_t}$ to be non-empty.

	\begin{figure}[t]
		\centering
		\begin{tikzpicture}
	\draw[teal, line width = 0.3mm, dotted] (-1, 1.4) rectangle (7, -0.5);
	\node at (-0.75, 1.1) {\bf\textcolor{teal}{$G$}};

	\node[place] (a) at (0,0) {$c_s$};
	\node[place] (b) at (1,0) {$p_s$};
	\node[transition, inner sep = 2pt] (c) at (0.5, 1) {$s_\textup{sim}$};
	\path[->] (c) edge[] (a);
	\path[->] (c) edge[] (b);
	\path[-] (c) edge[out = 0, in = 270, red] node[right, red] {\small$\overset{{\tiny?}}{=}\!\!\:0$} (1.25, 1.35);

	\node[place] (con) at (2.5,0) {$c_t$};
	\node[place] (pro) at (3.5,0) {$p_t$};
	\node[transition, inner sep = 2pt] (t0) at (3, 1) {$\sim$};
	\path[->] (t0) edge[] (con);
	\path[->] (t0) edge[] (pro);
	\path[-] (t0) edge[out = 0, in = 270, red] node[right, red] {\small$\overset{{\tiny?}}{=}\!\!\:0$} (3.75, 1.35);

	\node[place] (x) at (5,0) {$c_u$};
	\node[place] (y) at (6,0) {$p_u$};
	\node[transition, inner sep = 2pt] (z) at (5.5, 1) {$u_\textup{sim}$};
	\path[->] (z) edge[] (x);
	\path[->] (z) edge[] (y);
	\path[-] (z) edge[out = 0, in = 270, red] node[right, red] {\small$\overset{{\tiny?}}{=}\!\!\:0$} (6.25, 1.35);

	\node[transition, inner sep = 2pt] (t1) at (1.5, -1.5) {$\con$};
	\node[transition, inner sep = 2pt] (t2) at (4.5, -1.5) {$\pro$};
	\path[->] (con) edge[out = 210, in = 90] (t1);
	\path[->] (pro) edge[out = 330, in = 90] (t2);

	\node[place] (p) at (3, -1) {$a$};
	\node[place] (q) at (3, -2) {$b$};
	\path[]
	(p) edge[->] (t1)
	(t1) edge[-, red] node[below, red] {\small$\overset{{\tiny?}}{=}\!\!\:0$} (q)
	(t2) edge[->] (p)
	(t2) edge[->] node[below] {$2$} (q)
	(t2) edge[-, red] node[above, pos = 0.25, rotate = -45, red] {\small$\overset{{\tiny?}}{=}\!\!\:0$} (con)
	;

	\node[place] (op) at (-3.5, 1) {$a$};
	\node[transition, inner sep = 2pt] (t) at (-3.5, -0.5) {$t$};
	\node[place] (oq) at (-3.5, -2) {$b$};
	\path[]
	(op) edge[->, out=255, in=110] (t)
	(t) edge[->, out=75, in=290] (op)
	(t) edge[->, out = 250, in = 105] node[left] {$2$} (oq)
	(t) edge[-, out = 290, in = 75, red] node[right, red] {\small$\overset{{\tiny?}}{=}\!\!\:0$} (oq)
	;

	\node[transition, inner sep = 2pt] (t) at (-4.75, -0.5) {$s$};
	\node[transition, inner sep = 2pt] (t) at (-2.25, -0.5) {$u$};
\end{tikzpicture}
		\caption{
		Suppose there is a Petri net with zero tests with transitions $s$, $t$, and $u$.
		Left: part of the Petri net with zero tests concerning the transition $t$.
		The consumptions and productions of $s$ and $u$ are not shown for simplicity.
		This Petri net is \emph{not} acyclic since $t$ both consumes 1 token from and produces 1 token to $a$.
		Right: part of the equivalent acyclic Petri net with zero tests.
		Places in $G$ are shown, as well as transitions $s_\textup{sim}$, $\sim$, and $u_\textup{sim}$ for choosing the next transition to be fired.
		Since the consumptions and productions of $s$ and $u$ are not shown for simplicity, we also omit the corresponding $s_\textup{con}$, $s_\textup{pro}$, $u_\textup{con}$, and $u_\textup{pro}$.
		Importantly, zero-test edges between $\con$ and $b$, and between $\pro$ and $c_t$ are not subject to the acyclicity restriction.
		Zero-test edges incident to $s_\textup{sim}$, $\sim$, and $u_\textup{sim}$ indicate that all places in $G$ are zero-tested.
		}
		\label{fig:transition-controller}
	\end{figure}
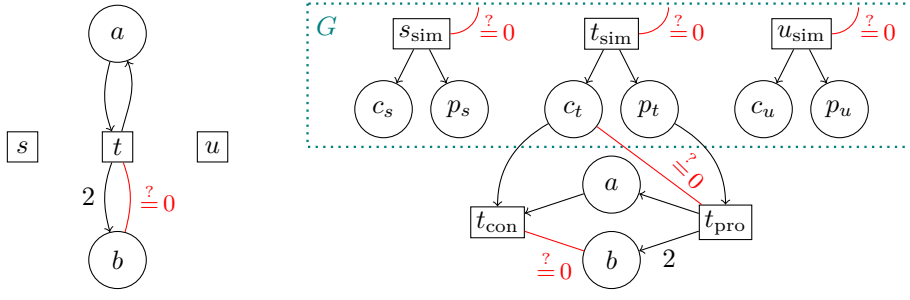

	Given a marking $\vec{v}$ over $P$, define $\vec{v}'$ over $P'$ such that $\vec{v}'[p] = \vec{v}[p]$ for all $p \in P$ and $\vec{v}'[q] = 0$ for all $q \in G$.
	It follows that $\run{\vec{m}}{*}{\vec{n}}$ in $\Pp$ if and only if $\run{\vec{m}'}{*}{\vec{n}'}$ in $\Zz$. 
	Indeed, runs in $\Pp$ have equivalent runs in $\Zz$, where each firing of a transition $t$ is replaced with the firing of transitions $\sim$, then $\con$, then $\pro$. 
	Conversely, as previously detailed, runs in $\Zz$ must fire $\sim$, $\con$, and $\pro$ successively for some transition $t \in T$.

	It remains to observe that $\Zz$ is acyclic.
	Consider the following ordering: choice simulation transitions (such as $\sim$) occur before places in $G$, which occur before production transitions (such as $\pro$), which occur before places in $P$, which occur before consumption transitions (such as $\con$).
\end{proof}

\begin{lemma}\label{lem:petri-simulate-zero-tests}
	The reachability problem for \zapn{s} is reducible in logarithmic space to the reachability problem for \rapn{s}.
\end{lemma}
\begin{proof}[Proof of~\cref{lem:petri-simulate-zero-tests}]
	Via leveraging the reachability objective, the idea is to add a copy of each place that will make sure zero tests are simulated faithfully.
	
	Let $\Zz = (P, T, F, Z)$ be an \zapn{}.
	We will construct an \rapn{} $\Rr = (P', T', F', R)$.
	For each place $p \in P$ we will add a copy place $c_p$, so $P' = \set{ p, c_p : p \in P }$.
	For each transition $t \in T$, there will be a corresponding transition $t' \in T'$ with the following behaviour.
	Firstly, $t'$ will mimic the token consumption and token production between the original places and their copies, so for every place $p$, $\pre{t'}[p] = \pre{t'}[c_p] = \pre{t}[p]$ and $\post{t'}[p] = \post{t'}[c_p] = \post{t}[p]$.
	Secondly, suppose a place $p \in P$ is zero-tested by $t \in T$, \ie $p \in Z(t)$.
	Then $t' \in T'$ will reset $p \in P'$ but not the copy $c_p \in P'$.
	Note that none of the copy places are ever reset.
	
	Given the initial marking $\vec{m}$ and target marking $\vec{n}$ over $P$, we define $\vec{m}'$ and $\vec{n}'$ over $P'$ so that $\vec{m}[p] = \vec{m}'[p] = \vec{m}'[c_p]$ and $\vec{n}[p] = \vec{n}'[p] = \vec{n}'[c_p]$.
	In other words, the markings over $P'$ allocate the same number of tokens to the copy places as their original counterparts.
	Suppose $\run{\vec{m}'}{*}{\vec{n}'}$ in $\Rr$.
	Then the invariant $\vec{m}'[c_p] - \vec{m}'[p] \leq \vec{n}'[c_p] - \vec{n}'[p]$ holds for all $p \in P$. 
	This inequality is strict only if at some point during the run, a transition is fired that resets a non-empty place.
	Therefore, $\run{\vec{m}}{*}{\vec{n}}$ in $\Zz$ if and only if $\run{\vec{m}'}{*}{\vec{n}'}$ in $\Rr$. 
	Indeed, a zero test on $p$ succeeds in $\Zz$ if and only if its corresponding reset has no effect; this occurs when $p$ and $c_p$ are empty.
	To conclude, it is clear that this is a logarithmic-space reduction and that the acyclicity of the consumption and production arcs between places and transitions is preserved in $\Rr$.
\end{proof}

\begin{remark*}
	\cref{lem:petri-simulate-transitions} does not hold with the workflow properties but~\cref{lem:petri-simulate-zero-tests} does; neither holds for the coverability objective.
\end{remark*}

\begin{proof}[Proof of~\cref{thm:petri-reachability}]
	Combine \cref{lem:petri-simulate-transitions}, \cref{lem:petri-simulate-zero-tests}, and the fact that reachability in Petri nets with zero tests is undecidable.
\end{proof}

\bibliography{bib}

\begin{thebibliography}{10}

\bibitem{ArakiK76}
Toshiro Araki and Tadao Kasami.
\newblock Some decision problems related to the reachability problem for
  {Petri} nets.
\newblock {\em Theor. Comput. Sci.}, 3(1):85--104, 1976.
\newblock \href {https://doi.org/10.1016/0304-3975(76)90067-0}
  {\path{doi:10.1016/0304-3975(76)90067-0}}.

\bibitem{BlondinHMR21}
Michael Blondin, Christoph Haase, Filip Mazowiecki, and Mikhail~A. Raskin.
\newblock Affine extensions of integer vector addition systems with states.
\newblock {\em Log. Methods Comput. Sci.}, 17(3), 2021.
\newblock \href {https://doi.org/10.46298/lmcs-17(3:1)2021}
  {\path{doi:10.46298/lmcs-17(3:1)2021}}.

\bibitem{BlondinMO22LICS}
Michael Blondin, Filip Mazowiecki, and Philip Offtermatt.
\newblock The complexity of soundness in workflow nets.
\newblock In Christel Baier and Dana Fisman, editors, {\em {LICS} '22: 37th
  Annual {ACM/IEEE} Symposium on Logic in Computer Science, Haifa, Israel,
  August 2 - 5, 2022}, pages 20:1--20:13. {ACM}, 2022.
\newblock \href {https://doi.org/10.1145/3531130.3533341}
  {\path{doi:10.1145/3531130.3533341}}.

\bibitem{BlondinMO22CAV}
Michael Blondin, Filip Mazowiecki, and Philip Offtermatt.
\newblock Verifying generalised and structural soundness of workflow nets via
  relaxations.
\newblock In Sharon Shoham and Yakir Vizel, editors, {\em Computer Aided
  Verification - 34th International Conference, {CAV} 2022, Haifa, Israel,
  August 7-10, 2022, Proceedings, Part {II}}, volume 13372 of {\em Lecture
  Notes in Computer Science}, pages 468--489. Springer, 2022.
\newblock \href {https://doi.org/10.1007/978-3-031-13188-2_23}
  {\path{doi:10.1007/978-3-031-13188-2_23}}.

\bibitem{BlondinR21}
Michael Blondin and Mikhail~A. Raskin.
\newblock The complexity of reachability in affine vector addition systems with
  states.
\newblock {\em Log. Methods Comput. Sci.}, 17(3), 2021.
\newblock \href {https://doi.org/10.46298/lmcs-17(3:3)2021}
  {\path{doi:10.46298/lmcs-17(3:3)2021}}.

\bibitem{ChistikovHH18}
Dmitry Chistikov, Christoph Haase, and Simon Halfon.
\newblock Context-free commutative grammars with integer counters and resets.
\newblock {\em Theor. Comput. Sci.}, 735:147--161, 2018.
\newblock Special issue for {RP~2014}.
\newblock \href {https://doi.org/10.1016/j.tcs.2016.06.017}
  {\path{doi:10.1016/j.tcs.2016.06.017}}.

\bibitem{CzerwinskiLLLM21}
Wojciech Czerwinski, Slawomir Lasota, Ranko Lazic, J{\'{e}}r{\^{o}}me Leroux,
  and Filip Mazowiecki.
\newblock The reachability problem for {Petri} nets is not elementary.
\newblock {\em J. {ACM}}, 68(1):7:1--7:28, 2021.
\newblock \href {https://doi.org/10.1145/3422822} {\path{doi:10.1145/3422822}}.

\bibitem{CzerwinskiO21}
Wojciech Czerwinski and Lukasz Orlikowski.
\newblock Reachability in vector addition systems is {Ackermann}-complete.
\newblock In {\em 62nd {IEEE} Annual Symposium on Foundations of Computer
  Science, {FOCS} 2021, Denver, CO, USA, February 7-10, 2022}, pages
  1229--1240. {IEEE}, 2021.
\newblock \href {https://doi.org/10.1109/FOCS52979.2021.00120}
  {\path{doi:10.1109/FOCS52979.2021.00120}}.

\bibitem{DufourdFS98}
Catherine Dufourd, Alain Finkel, and Philippe Schnoebelen.
\newblock Reset nets between decidability and undecidability.
\newblock In Kim~Guldstrand Larsen, Sven Skyum, and Glynn Winskel, editors,
  {\em Automata, Languages and Programming, 25th International Colloquium,
  ICALP'98, Aalborg, Denmark, July 13-17, 1998, Proceedings}, volume 1443 of
  {\em Lecture Notes in Computer Science}, pages 103--115. Springer, 1998.
\newblock \href {https://doi.org/10.1007/BFb0055044}
  {\path{doi:10.1007/BFb0055044}}.

\bibitem{FahlandFJKLVW09}
Dirk Fahland, C{\'{e}}dric Favre, Barbara Jobstmann, Jana Koehler, Niels
  Lohmann, Hagen V{\"{o}}lzer, and Karsten Wolf.
\newblock Instantaneous soundness checking of industrial business process
  models.
\newblock In Umeshwar Dayal, Johann Eder, Jana Koehler, and Hajo~A. Reijers,
  editors, {\em Business Process Management, 7th International Conference,
  {BPM} 2009, Ulm, Germany, September 8-10, 2009. Proceedings}, volume 5701 of
  {\em Lecture Notes in Computer Science}, pages 278--293. Springer, 2009.
\newblock \href {https://doi.org/10.1007/978-3-642-03848-8_19}
  {\path{doi:10.1007/978-3-642-03848-8_19}}.

\bibitem{FigueiraFSS11}
Diego Figueira, Santiago Figueira, Sylvain Schmitz, and Philippe Schnoebelen.
\newblock Ackermannian and primitive-recursive bounds with {Dickson's} lemma.
\newblock In {\em Proceedings of the 26th Annual {IEEE} Symposium on Logic in
  Computer Science, {LICS} 2011, June 21-24, 2011, Toronto, Ontario, Canada},
  pages 269--278. {IEEE} Computer Society, 2011.
\newblock \href {https://doi.org/10.1109/LICS.2011.39}
  {\path{doi:10.1109/LICS.2011.39}}.

\bibitem{FinkelMP04}
Alain Finkel, Pierre McKenzie, and Claudine Picaronny.
\newblock A well-structured framework for analysing {Petri} net extensions.
\newblock {\em Inf. Comput.}, 195(1-2):1--29, 2004.
\newblock \href {https://doi.org/10.1016/j.ic.2004.01.005}
  {\path{doi:10.1016/j.ic.2004.01.005}}.

\bibitem{HaaseKOW09}
Christoph Haase, Stephan Kreutzer, Jo{\"{e}}l Ouaknine, and James Worrell.
\newblock Reachability in succinct and parametric one-counter automata.
\newblock In Mario Bravetti and Gianluigi Zavattaro, editors, {\em {CONCUR}
  2009 - Concurrency Theory, 20th International Conference, {CONCUR} 2009,
  Bologna, Italy, September 1-4, 2009. Proceedings}, volume 5710 of {\em
  Lecture Notes in Computer Science}, pages 369--383. Springer, 2009.
\newblock \href {https://doi.org/10.1007/978-3-642-04081-8_25}
  {\path{doi:10.1007/978-3-642-04081-8_25}}.

\bibitem{HiraishiI88}
Kunihiko Hiraishi and Atsunobu Ichikawa.
\newblock A class of {Petri} nets that a necessary and sufficient condition for
  reachability is obtainable.
\newblock {\em Transactions of the Society of Instrument and Control
  Engineers}, 24(6):635--640, 1988.
\newblock \href {https://doi.org/10.9746/sicetr1965.24.635}
  {\path{doi:10.9746/sicetr1965.24.635}}.

\bibitem{0001KW14}
Alexander Kaiser, Daniel Kroening, and Thomas Wahl.
\newblock A widening approach to multithreaded program verification.
\newblock {\em {ACM} Trans. Program. Lang. Syst.}, 36(4):14:1--14:29, 2014.
\newblock \href {https://doi.org/10.1145/2629608} {\path{doi:10.1145/2629608}}.

\bibitem{KunnemannMSSW23}
Marvin K\"{u}nnemann, Filip Mazowiecki, Lia Sch\"{u}tze, Henry Sinclair-Banks,
  and Karol W\k{e}grzycki.
\newblock {Coverability in VASS Revisited: Improving Rackoff’s Bound to
  Obtain Conditional Optimality}.
\newblock In Kousha Etessami, Uriel Feige, and Gabriele Puppis, editors, {\em
  50th International Colloquium on Automata, Languages, and Programming (ICALP
  2023)}, volume 261 of {\em Leibniz International Proceedings in Informatics
  (LIPIcs)}, pages 131:1--131:20, Dagstuhl, Germany, 2023. Schloss Dagstuhl --
  Leibniz-Zentrum f{\"u}r Informatik.
\newblock URL: \url{https://drops.dagstuhl.de/opus/volltexte/2023/18183}, \href
  {https://doi.org/10.4230/LIPIcs.ICALP.2023.131}
  {\path{doi:10.4230/LIPIcs.ICALP.2023.131}}.

\bibitem{Leroux21}
J{\'{e}}r{\^{o}}me Leroux.
\newblock The reachability problem for {Petri} nets is not primitive recursive.
\newblock In {\em 62nd {IEEE} Annual Symposium on Foundations of Computer
  Science, {FOCS} 2021, Denver, CO, USA, February 7-10, 2022}, pages
  1241--1252. {IEEE}, 2021.
\newblock \href {https://doi.org/10.1109/FOCS52979.2021.00121}
  {\path{doi:10.1109/FOCS52979.2021.00121}}.

\bibitem{LerouxS19}
J{\'{e}}r{\^{o}}me Leroux and Sylvain Schmitz.
\newblock Reachability in vector addition systems is primitive-recursive in
  fixed dimension.
\newblock In {\em 34th Annual {ACM/IEEE} Symposium on Logic in Computer
  Science, {LICS} 2019, Vancouver, BC, Canada, June 24-27, 2019}, pages 1--13.
  {IEEE}, 2019.
\newblock \href {https://doi.org/10.1109/LICS.2019.8785796}
  {\path{doi:10.1109/LICS.2019.8785796}}.

\bibitem{LiSGGZ11}
Duan Li, Xiaoling Sun, Jianjun Gao, Shenshen Gu, and Xiaojin Zheng.
\newblock Reachability determination in acyclic {Petri} nets by cell
  enumeration approach.
\newblock {\em Autom.}, 47(9):2094--2098, 2011.
\newblock \href {https://doi.org/10.1016/j.automatica.2011.06.017}
  {\path{doi:10.1016/j.automatica.2011.06.017}}.

\bibitem{Lipton76}
Richard~Jay Lipton.
\newblock {\em The reachability problem requires exponential space}.
\newblock Research Report. Department of Computer Science, Yale University,
  1976.
\newblock URL: \url{http://www.cs.yale.edu/publications/techreports/tr63.pdf}.

\bibitem{Minsky67}
Marvin~L. Minsky.
\newblock {\em Computation: Finite and Infinite Machines}.
\newblock Prentice-Hall, Inc., 1967.

\bibitem{Petri75}
Carl~Adam Petri.
\newblock Introduction to general net theory.
\newblock In Wilfried Brauer, editor, {\em Net Theory and Applications,
  Proceedings of the Advanced Course on General Net Theory of Processes and
  Systems, Hamburg, Germany, October 8-19, 1979}, volume~84 of {\em Lecture
  Notes in Computer Science}, pages 1--19. Springer, 1979.
\newblock \href {https://doi.org/10.1007/3-540-10001-6_21}
  {\path{doi:10.1007/3-540-10001-6_21}}.

\bibitem{DBLP:journals/tcs/Rackoff78}
Charles Rackoff.
\newblock The covering and boundedness problems for vector addition systems.
\newblock {\em Theor. Comput. Sci.}, 6:223--231, 1978.
\newblock \href {https://doi.org/10.1016/0304-3975(78)90036-1}
  {\path{doi:10.1016/0304-3975(78)90036-1}}.

\bibitem{Schnoebelen10}
Philippe Schnoebelen.
\newblock Revisiting {Ackermann}-hardness for lossy counter machines and reset
  {Petri} nets.
\newblock In Petr Hlinen{\'{y}} and Anton{\'{\i}}n Kucera, editors, {\em
  Mathematical Foundations of Computer Science 2010, 35th International
  Symposium, {MFCS} 2010, Brno, Czech Republic, August 23-27, 2010.
  Proceedings}, volume 6281 of {\em Lecture Notes in Computer Science}, pages
  616--628. Springer, 2010.
\newblock \href {https://doi.org/10.1007/978-3-642-15155-2_54}
  {\path{doi:10.1007/978-3-642-15155-2_54}}.

\bibitem{TipleaBC15}
Ferucio~Laurentiu Tiplea, Corina Bocaneala, and Raluca Chirosca.
\newblock On the complexity of deciding soundness of acyclic workflow nets.
\newblock {\em {IEEE} Trans. Syst. Man Cybern. Syst.}, 45(9):1292--1298, 2015.
\newblock \href {https://doi.org/10.1109/TSMC.2015.2394735}
  {\path{doi:10.1109/TSMC.2015.2394735}}.

\bibitem{TipleaM05}
Ferucio~Laurentiu Tiplea and Dan~C. Marinescu.
\newblock Structural soundness of workflow nets is decidable.
\newblock {\em Inf. Process. Lett.}, 96(2):54--58, 2005.
\newblock \href {https://doi.org/10.1016/j.ipl.2005.06.002}
  {\path{doi:10.1016/j.ipl.2005.06.002}}.

\bibitem{Valk78}
R{\"{u}}diger Valk.
\newblock Self-modifying nets, a natural extension of {Petri} nets.
\newblock In Giorgio Ausiello and Corrado B{\"{o}}hm, editors, {\em Automata,
  Languages and Programming, Fifth Colloquium, Udine, Italy, July 17-21, 1978,
  Proceedings}, volume~62 of {\em Lecture Notes in Computer Science}, pages
  464--476. Springer, 1978.
\newblock \href {https://doi.org/10.1007/3-540-08860-1_35}
  {\path{doi:10.1007/3-540-08860-1_35}}.

\bibitem{Aalst98}
Wil M.~P. van~der Aalst.
\newblock The application of {Petri} nets to workflow management.
\newblock {\em J. Circuits Syst. Comput.}, 8(1):21--66, 1998.
\newblock \href {https://doi.org/10.1142/S0218126698000043}
  {\path{doi:10.1142/S0218126698000043}}.

\bibitem{AalstHHSVVW09}
Wil M.~P. van~der Aalst, Kees~M. van Hee, Arthur H.~M. ter Hofstede, Natalia
  Sidorova, H.~M.~W. Verbeek, Marc Voorhoeve, and Moe~Thandar Wynn.
\newblock Soundness of workflow nets with reset arcs.
\newblock {\em Trans. Petri Nets Other Model. Concurr.}, 3:50--70, 2009.
\newblock \href {https://doi.org/10.1007/978-3-642-04856-2_3}
  {\path{doi:10.1007/978-3-642-04856-2_3}}.

\bibitem{AalstHHSVVW11}
Wil M.~P. van~der Aalst, Kees~M. van Hee, Arthur H.~M. ter Hofstede, Natalia
  Sidorova, H.~M.~W. Verbeek, Marc Voorhoeve, and Moe~Thandar Wynn.
\newblock Soundness of workflow nets: classification, decidability, and
  analysis.
\newblock {\em Formal Aspects Comput.}, 23(3):333--363, 2011.
\newblock \href {https://doi.org/10.1007/s00165-010-0161-4}
  {\path{doi:10.1007/s00165-010-0161-4}}.

\end{thebibliography}

\begin{appendix}
    \section{Missing proofs of~\cref{sec:petri-coverability}}
\label{app:petri-coverability}

\coverabilityholds*
\begin{proof}
	First, suppose that there is some firing sequence $\pi$ such that $\run{\vec{m}}{\pi}{\vec{n}'}$ in $\Pp$ where $\vec{n}' \geq \vec{n}$. 
	By firing the same sequence of transitions in $\Nn$, under its semantics, then $\vec{n}$ will be covered from $\vec{m}$.
	This holds because if $\run{\vec{p}}{t}{\vec{q}}$ in $\Pp$, then $\run{\vec{p}}{t}{\vec{q}'}$ in $\Nn$ where $\vec{q}' \geq \vec{q}$.
	Clearly if $t$ is not generating from $\vec{p}$, then $\vec{q}' = \vec{q}$, so consider the case when $t$ is generating from $\vec{q}$.
	Indeed, for a place $p$ such that $p \notin \reset{t}$ and $\post{t}[p] > 0$, then $\vec{q}'[p] = \omega > \vec{q}[p] = \vec{p}[p] - \pre{t}[p] + \post{t}[p]$.
	And, for a place $p$ such that $p \in \reset{t}$, then $\vec{q}'[p] = \post{t}[p] = \vec{q}[p]$.
	It follows that $\run{\vec{m}}{\pi}{\vec{n}''}$ in $\Nn$, where $\vec{n}'' \geq \vec{n}' \geq \vec{n}$.
	
	Now, suppose that there is some firing sequence $\rho$ such that $\run{\vec{m}}{\rho}{\vec{n}'}$ in $\Nn$, where $\vec{n}' \geq \vec{n}$.
	We actually show a stronger statement than what is necessary.
	Suppose the target marking may have places that contain $\omega$ tokens, so $\vec{n} \in \Nw^n$.
	Consider any $\vec{n}^* \in \N^n$ such that $\vec{n}^*[p] \in \N$ if $\vec{n}[p] = \omega$, and $\vec{n}^*[p] = \vec{n}[p]$ if $\vec{n}[p] \in \N$ (for each place $p$). So every $\omega$ is replaced with an arbitrary natural number.
	We will argue, by induction on the length of the firing sequence $\rho$, that $\vec{n}^*$ is coverable from $\vec{m}$ in $\Pp$ if $\vec{n}$ is coverable from $\vec{m}$ in $\Nn$.
	
	The induction base is trivial.
	If the length of $\rho$ is zero, thus $\vec{n}^* = \vec{n}' = \vec{m}$, thus $\vec{n}^*$ is coverable from $\vec{m}$ in $\Pp$.
	For the inductive step, let $\rho = \sigma t$, so assume that $\vec{m} \overset{\sigma}{\rightarrow} \vec{x} \overset{t}{\rightarrow} \vec{n}'$ in $\Nn$, here $t$ is the last transition of $\rho$.
	Consider any $\vec{x}^* \in \N^n$ such that $\vec{x}^*[p] \in \N$ if $\vec{x}[p] = \omega$, and $\vec{x}^*[p] = \vec{x}[p]$ if $\vec{x}[p] \in \N$ (for each place $p$).
	Since $\run{\vec{m}}{\sigma}{\vec{x}}$ in $\Nn$, it is true that $\vec{x}$ is coverable from $\vec{m}$ in $\Nn$, and since $\sigma$ is (one transition) shorter than $\rho$, then by the inductive assumption, we know that $\vec{x}^*$ is coverable from $\vec{m}$ in $\Pp$.

	Recall that $\vec{n}^* \in \N$; let $m$ by the greatest number of tokens in a place in $\vec{n}^*$ and let $k$ be the maximal number of tokens consumed by a single transition in $\Pp$.
	We will split into two cases based on transition $t$ and marking $\vec{x}$. 
	Case one is that $t$ is a non-generating transition from $\vec{x}$, and case two is that $t$ is generating from $\vec{x}$.

	In case one, if $\vec{n}'[p] = \omega$ for some place $p$, then $\vec{x}[p] = \omega$ must be true.
	Define $\vec{x}^*$ so that $\vec{x}^*[p] = m+k$ if $\vec{x}[p] = \omega$, and $\vec{x}^*[p] = \vec{x}[p]$ if $\vec{x}[p] \in \N$ (for each place $p$).
	Consider the marking $\vec{v}$ reached after firing $t$ from $\vec{x}^*$ in $\Pp$.
	By definition of $m$ and $k$, it is clear that $\vec{v} \geq \vec{n}^*$.
	This concludes case one.

	In case two, $t$ is generating from $\vec{x}$.
	In that case there may exist a place $p$ such that $\vec{n'}[p] = \omega$, but $\vec{x}[p] \in \N$.
	This occurs when $t$ increases the number of tokens in $p$.
	If, in $\Pp$, $t$ is fired exactly $m$ times, then the number of tokens in the place $p$ will be at least $M$ and the number of tokens consumed from other places will be at most $mk$ (unless the place in question is reset).
	Thus, we define $\vec{x}^*$ such that $\vec{x}^*[p] = m + mk$ if $\vec{x}[p] = \omega$, and $\vec{x}^*[p] = \vec{x}[p]$ if $\vec{x}[p] \in \N$.
	This time, consider the marking $\vec{v}$ reached after firing $t$ repeatedly $m$ times from $\vec{x}^*$ in $\Pp$, so $\run{\vec{x}^*}{t^m}{\vec{v}}$.
	By definition of $m$ and $k$, it follows that $\vec{v}[p] = \vec{n}^*[p]$ for each place $p$ such that $\post{t}[p] = 0$ and $\vec{v}[p] \geq \vec{n}^*[p]$ for each place $p$ such that $\post{t}[p] > 0$. 
	Hence $\vec{v} \geq \vec{n}^*$.
	This concludes case two.
\end{proof}

\exponentialmarkings*
\begin{proof}
	Since $\Pp$ is acyclic, then $\Nn$ is clearly acyclic.
	We can assume that places $p_1, \ldots, p_n$ of $\Nn$ are ordered such that each transition that consumes tokens from some place $p_i$ does not produce tokens to a place $p_j$ for $j < i$. 

	We define the \emph{weight} of a marking, $\mathrm{weight}(\vec{v}) \coloneqq \sum_{i \in [1, n]: \vec{v}[p_i] \in \N} k^{n-i+1} \cdot \vec{v}[p]$.
	Note that $\mathrm{weight}(\vec{m}) = k^n\cdot\vec{m}[p_1] + \ldots + k^1\cdot\vec{m}[p_n] \leq C \cdot k^n$. 
	We will now show that firing a transition does not increase the weight of a marking.
	A transition $t$ in $\Nn$ must either consume tokens from some place, or it is a generating transition.
	In the first case, if a transition $t$ consumes some number of tokens from a place $p_i$, then it may only produce tokens to places $p_j$ for $j > i$.
	The number of tokens produced by $t$ is bounded above by $k$, as $t$ produces at most $k$ tokens.
	The consumption of a token from $p_i$ at least decreases the weight of the marking reached by $k^{n-i+1}$ and the production of at most $k$ tokens to some $p_1, \ldots, p_{i-1}$ can increase the weight of the marking reached by at most $k \cdot k^{n-i}$.
	Hence, the weight of the current marking cannot increase by firing $t$.
	In the second case, firing a transition $t$ may only introduce $\omega$'s in some places.
	Since places containing $\omega$ tokens do not contribute to the weight of the marking, the weight of the new marking does not increase.

	Since $\run{\vec{m}}{*}{\vec{v}}$ in $\Nn$, we therefore know that $\mathrm{weight}(\vec{v}) \leq \mathrm{weight}(\vec{m}) \leq C \cdot k^n$.
	Clearly $\mathrm{weight}(\vec{v}) \geq \norm{\vec{v}}$, so $\norm{\vec{v}} \leq C \cdot k^n$, as required.
\end{proof}

\section{Missing proofs of~\cref{sec:workflow-coverability}}
\label{app:workflow-coverability}

\binarytransitions*
\begin{proof}
	Without loss of generality, we will assume that each of the binary arcs has a weight that is a power of two.

	Let $\Bb = (P, T, F, R)$ be the given \rawn{} with binary-encoded arc weights.
	We will first construct the (unary-encoded) \rawn{} $\Uu = (P', T', F', R')$.
	There will be a place $p' \in P'$ for each place $p \in P$ as well as a series of additional places that are reserved for transition gadgets.
	The transitions $T'$ of $\Uu$ will come from a series of transition gadgets used to simulate the transitions of $\Bb$.

	Consider a transition $t$ that consumes $2^a$ tokens from place $p$ and produces $2^b$ tokens into place $q$.
	The weak simulation gadget $G_t$, depicted in~\cref{fig:binary-transition}, has the property that, prior to the \emph{key transition} $s_t$ being fired, $2^a$ tokens must have been consumed from $p'$.
	It is important to note that if $p$ is reset in $\Bb$, then all places, except $q'$, in $G_t$ are reset. 
	Similarly, if $q$ is reset in $\Bb$, then all places, except $p'$, in $G_t$ are reset.
	\begin{figure}[ht!]
		\centering
		\begin{tikzpicture}
	\node[place, label=below:$p$] (p) at (-12, -0.75) {};
	\node[place, label=below:$q$] (q) at (-9, -0.75) {};
	\node[transition] (t) at (-10.5, -0.75) {$t$};
	\path[->] 
		(p) edge[] node[above] {$2^a$} (t) 
		(t) edge[] node[above] {$2^b$} (q);

	\node[place, label=left:$p'$] (pu) at (-6, 0) {};
	\node[place] (c1) at (-4, 0) {};
	\node (cd) at (-2, 0) {\small$\cdots$};
				
	\node[place] (ck) at (0, 0) {};
	\node[transition] (tc1) at (-5, 0) {};
	\node[transition] (tc2) at (-3, 0) {};
	\node[transition] (tc3) at (-1, 0) {};
	\path[->] 
		(pu.-11) edge[] (tc1.-160) 
		(pu.11) edge[] (tc1.160)
		(tc1) edge[] (c1)
		(c1.-11) edge[] (tc2.-160) 
		(c1.11) edge[] (tc2.160)
		(tc2) edge[] (cd)
		(cd.-11) edge[] (tc3.-160) 
		(cd.11) edge[] (tc3.160)
		(tc3) edge[] (ck);

                \draw [decorate, decoration={brace, raise=0.75ex}] (tc1.west |- c1.north) -- (tc3.east |- c1.north) node [midway, above, yshift=1.25ex] {$a$ transitions};

	\node[transition] (tu) at (0, -0.75) {$s_t$};

	\node[place] (pk) at (0, -1.5) {};
	\node[place] (p1) at (-2, -1.5) {};
	\node (pd) at (-4, -1.5) {$\small\cdots$};
			
	\node[place, label=left:$q'$] (qu) at (-6, -1.5) {};
	\node[transition] (tp1) at (-1, -1.5) {};
	\node[transition] (tp2) at (-3, -1.5) {};
	\node[transition] (tp3) at (-5, -1.5) {};

	\path[->]
		(pk) edge[] (tp1)
		(tp1.160) edge[] (p1.11)
		(tp1.-160) edge[] (p1.-11)
		(p1) edge[] (tp2)
		(tp2.160) edge[] (pd.11)
		(tp2.-160) edge[] (pd.-11)
		(pd) edge[] (tp3)
		(tp3.160) edge[] (qu.11)
		(tp3.-160) edge[] (qu.-11);

                \draw [decorate, decoration={brace, raise=0.75ex}] (tp1.east |- p1.south) -- (tp3.west |- p1.south) node [midway, below, yshift=-1.25ex] {$b$ transitions};

	\path[->] 
		(ck) edge[] (tu)
		(tu) edge[] (pk);

	\node at (-7, 0.75) {$G_t$};
\end{tikzpicture}
		\caption{The (weak) simulation of an example binary transition by a series unary transitions.
		Note the labelled \emph{key transition} $s_t$ in the transition gadget $G_t$.}
		\label{fig:binary-transition}
	\end{figure}
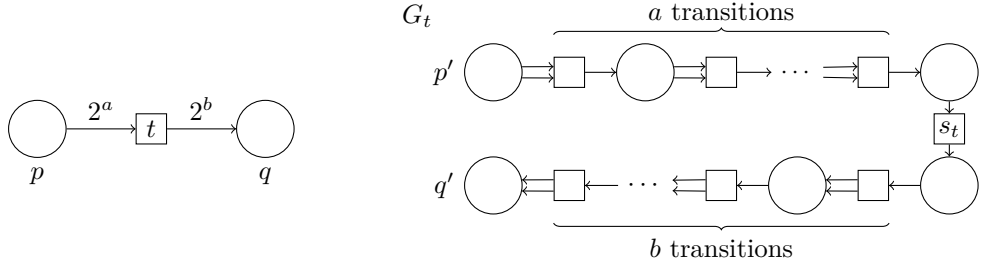

	The initial marking $\vec{x}$ and the target marking are set so that $\vec{x}[p'] = \vec{m}[p]$ and $\vec{y}[p'] = \vec{n}[p]$ for all $p \in P$; for all other places $q' \in P' \setminus P$, both $\vec{x}[q'] = 0$ and $\vec{y}[q'] = 0$.
	It remains to show that $I = (\Bb, \vec{m}, \vec{n})$ is positive if and only if $I' = (\Uu, \vec{x}, \vec{y})$ is positive.

	Suppose $I$ is positive.
	There exists a firing sequence $\pi$ such that $\run{\vec{m}}{\pi}{\vec{n}'}$ in $\Bb$ where $\vec{n}' \geq \vec{n}$.
	Consider the firing sequence $\rho$ in $\Uu$ that is obtained from $\pi$ by replacing the firing of $t$ with the firing of all transitions of $G_t$ in order an appropriate number of times so that $2^b$ tokens are produced in $q'$.
	It is clear that, from the initial marking $\vec{x}$, a marking $\vec{y}'$ will be reached such that $\vec{y}'[p'] = \vec{n}'[p]$ for all $p \in P$ and $\vec{y'}[q'] = 0$ for all other places $q'$.
	Since $\vec{n}' \geq \vec{n}$, it follows that $\vec{y}' \geq \vec{y}$.
	This implies that $I'$ is positive too.

	Now suppose $I'$ is positive.
	There exists a firing sequence $\rho = (t'_1, \ldots, t'_k)$ such that $\run{\vec{x}}{\rho}{\vec{y}'}$ in $\Uu$ where $\vec{y}' \geq \vec{y}$.
	Find the subsequence $J = (j_1, \ldots, j_\ell)$ of $(1, \ldots, k)$ such that for each $j \in J$, $t'_j$ is a key transition in some transition gadget $G_t$ (see~\cref{fig:binary-transition}), \ie there is some $t \in T$ for which $t'_j = s_t$.
	Consider the firing sequence $\pi = (t_1, \ldots, t_\ell)$ obtained by only firing the transition $t_j$ in $\Bb$ whenever the key transition $s_{t_j}$ is fired in $\Uu$.
	Indeed, suppose $t_j$ is a transition that consumes $2^a$ tokens from $p \in P$ and produces $2^b$ tokens in $q \in P$.
	If the key transition $s_{t_j}$ is fired in $\Uu$, then some $2^a$ many tokens must have previously been consumed from $p' \in P'$, and subsequently, the latter transitions of $G_t$ could be fired which would lead to the production of $2^b$ many tokens to $q' \in P'$.
	It may not be the case that all the transitions after $t'_j$ are fired in $G_t$ before a reset occurs.
	In that case, perhaps only some of the tokens are produced in $q'$ in $\Uu$ compared to $q$ in $\Bb$.
    Since our objective is coverability, we may only increase the number of tokens in $q$ compared to $q'$.
	Similarly, even before the firing of a key transition $s_t$, it could have been the case that some of the consumption transitions in $G_t$ were fired, producing some tokens in the places before $s_t$.
	Again, if these places are reset, this would only increase the number of tokens in the places $P$ compared to the places $P'$.
	Therefore, by firing $\pi$ from the initial marking $\vec{m}$, a marking $\vec{n'}$ can be reached, such that for every place $p \in P, \, \vec{n'}[p] \geq \vec{y'}[p'] \geq \vec{y}[p'] = \vec{n}[p]$.
	Hence $\vec{n}' \geq \vec{n}$, so $I$ is positive too.

	Finally, it is clear that $I'$ can be constructed in polynomial time, given that for each transition $t \in T$, the \rawn{} $\Uu$ only contains $C_1\cdot\log(\abs{P})$ many places and $C_2\cdot\log(\abs{T})$ many transitions, for some constants $C_1, C_2 \in \N$.
\end{proof}

\truexorfalse*
\begin{proof}
	For $\overline b_i$ to be non-empty $u_i^\bot$ must have been fired, and for $b_i$ to be non-empty $u_i^\top$ must have been fired.
	However, $u_i^\bot$ resets and does not produce any tokens to $b_i$ and $u_i^\top$ resets and does not produce any tokens to $\overline b_i$.
	No matter which of $u_i^\bot$ or $u_i^\top$ was fired most recently when reaching $\vec{v}$, either $\overline b_i$ or $b_i$ must be empty.
	The same argument applies to $\overline a_i$ and $a_i$ by considering the transitions $e_i^\bot$ and $e_i^\top$.
\end{proof}

\collectingtropies*
\begin{proof}
	There is only one transition, the satisfaction transition $s$, that produces tokens in $f$, and it only produces one token.
\end{proof}

\subparagraph*{Further remarks on good markings (\cref{def:good-marking}).}
For sake of intuition, consider the first condition with $i=1$: $2^k = g_1(\vec{v}) = \vec{v}[f] + 2^k\vec{v}[h_1] + 2^{k-1}\vec{v}[w_1] + \vec{v}[\overline b_1] + \vec{v}[b_1] + \vec{v}[d_{y_1}]$.
In the initial marking, $h_1$ contains a token and $w_1$, $\overline b_1$, $b_1$, and $d_{y_1}$ are empty, so the marking is good.
Upon the firing of the first transition $u_1^\bot$, a token is placed in $w_1$ and $2^{k-1}$ tokens are placed in $\overline b_1$. 

Notice at this point that, in order for a token to be eventually produced in the final place $f$, all clause places must be non-empty, which corresponds to all clauses being satisfied.
Observe that, from this marking, a token must be taken from $\overline b_1$ before $s$ can be fired: indeed, the dummy clause $(\overline y_1 \vee y_1)$ needs to be satisfied, which is indicated by $d_{y_1}$ containing a token.
So, either $\ell_{\overline y_1}$ or $\ell_{y_1}$ must be fired to place a token in $d_{y_1}$.
By~\cref{clm:true-xor-false}, if $\overline b_1$ is non-empty, then $b_1$ must be empty, so in this case $\ell_{\overline y_1}$ is fired, consuming a token from $\overline b_1$. 
The marking reached is still good because moving a token from $\overline b_1$ to $d_{y_1}$ maintains the balance.
From here, if $s$ is fired, a token is consumed from $d_{y_1}$ and a token is produced to $f$, again keeping the goodness condition satisfied.

Overall, the conditions on $g_i$ promise that the number of tokens in the universal gadgets are balanced with the number of tokens in the final place.
Similarly, the conditions on $g'_i$ promise that the number of tokens in the existential gadgets are balanced with the number of tokens in the final place.

Finally, consider $g'_i(\vec{v}) - g_i(\vec{v}) = \vec{v}[\overline a_i] + \vec{v}[a_i] + \vec{v}[d_{x_i}] + 2^{k-i}\cdot\vec{v}[v_i] - \vec{v}[\overline b_i] - \vec{v}[b_i] - \vec{v}[d_{y_i}]$.
If $\vec{v}$ is good, then $g'_i(\vec{v}) - g_i(\vec{v}) = 0$ for each $i \in [1, k]$.
These are balance conditions between pairs of existential gadgets and universal gadgets.
Intuitively, the total number of tokens in the places $\overline b_i$, $b_i$, and $d_{y_i}$ will be equal to the total number of tokens in (or soon to be in) the places $\overline a_i$, $a_i$, and $d_{x_i}$.
So, in a run that goes through good markings only, there is no way to run out of tokens for clauses $(\overline y_i \vee y_i)$ without running out of tokens for clauses $(\overline x_i \vee x_i)$.

We will now show that all $g_i$ and $g'_i$ are non-increasing with respect to firing a transition.

\invariantcantincrease*
\begin{proof}
	The proof will be split into case depending on which transition was fired, but first we give an outline of the proof.
	In each case, we will compare the number of tokens produces and the number of tokens consumed to their `relative value' to each function $g_1, g'_1, \ldots, g_k, g'_k$

	\begin{figure}[ht!]
		\centering
		
		{\renewcommand{\arraystretch}{0}%
		\begin{tabular}{@{\strut\hspace{\tabcolsep}}c|c c c c c c c c c c c}
		Places: & $f$ & $h_i$ & $w_i$ & $v_i$ & $\overline b_i$ & $b_i$ & $d_{y_i}$ & $\overline a_i$ & $a_i$ & $d_{x_i}$ \\ 
		\hline\multicolumn{1}{c|}{\rule{0pt}{2pt}}&\\\hline
		$u_i^\bot$ & $0$ & $-1$ & $1$ & $1$ & $2^{k-i}$ & $0$ & $0$ & $0$ & $0$ & $0$ \\
		$u_i^\top$ & $0$ & $0$ & $-1$ & $1$ & $0$ & $2^{k-i}$ & $0$ & $0$ & $0$ & $0$ \\ \hline
		$e_i^\bot$ & $0$ & $0$ & $0$ & $-1$ & $0$ & $0$ & $0$ & $2^{k-i}$ & $0$ & $0$ \\
		$e_i^\top$ & $0$ & $0$ & $0$ & $-1$ & $0$ & $0$ & $0$ & $0$ & $2^{k-i}$ & $0$ \\ \hline
		$\ell_{\overline y_i}$ & $0$ & $0$ & $0$ & $0$ & $-1$ & $0$ & $1$ & $0$ & $0$ & $0$ \\
		$\ell_{y_i}$ & $0$ & $0$ & $0$ & $0$ & $0$ & $-1$ & $1$ & $0$ & $0$ & $0$ \\
		$\ell_{\overline x_i}$ & $0$ & $0$ & $0$ & $0$ & $0$ & $0$ & $0$ & $-1$ & $0$ & $1$ \\
		$\ell_{x_i}$ & $0$ & $0$ & $0$ & $0$ & $0$ & $0$ & $0$ & $0$ & $-1$ & $1$ \\ \hline
		$s$ & $1$ & $0$ & $0$ & $0$ & $0$ & $0$ & $-1$ & $0$ & $0$ & $-1$ \\ 
		\hline\multicolumn{1}{c|}{\rule{0pt}{2pt}}&\\\hline
		$g_1, \ldots, g_{i-1}$ & $1$ & $0$ & $0$ & $0$ & $0$ & $0$ & $0$ & $0$ & $0$ & $0$ \\
		$g'_1, \ldots, g'_{i-1}$ & $1$ & $0$ & $0$ & $0$ & $0$ & $0$ & $0$ & $0$ & $0$ & $0$ \\ \hline
		$g_i$ & $1$ &$2^{k-i+1}$ & $2^{k-i}$ & $0$ & $1$ & $1$ & $1$ & $0$ & $0$ & $0$ \\
		$g'_i$ & $1$ & $2^{k-i+1}$ & $2^{k-i}$ & $2^{k-i}$ & $0$ & $0$ & $0$ & $1$ & $1$ & $1$ \\ \hline
		$g_{i+1}, \ldots, g_k$ & $1$ &$2^{k-i+1}$ & $2^{k-i}$ & $2^{k-i}$ & $0$ & $0$ & $0$ & $0$ & $0$ & $0$ \\
		$g'_{i+1}, \ldots, g'_k$ & $1$ & $2^{k-i+1}$ & $2^{k-i}$ & $2^{k-i}$ & $0$ & $0$ & $0$ & $0$ & $0$ & $0$ 
		\end{tabular}
		}
	
		\caption{The upper section of this table details the number of tokens that are consumed (the negative entries) or produced (the positive entries) by each of the transitions $u_i^\bot$, $u_i^\top$, $e_i^\bot$, $e_i^\top$, $\ell_{\overline y_i}$, $\ell_{y_i}$, $\ell_{\overline x_i}$, $\ell_{x_i}$, and $s$. 
		The lower section of this table details the relative value of each of these places, that are just the coefficients in the functions $g_1, g'_1, \ldots, g_k, g'_k$.}
		\label{fig:dot-product}
	\end{figure}

	We show that (by design) the number of tokens produced to each place balances out, in the best case, with the coefficient that place has in the functions $g_1, g'_1, \ldots, g_k, g'_k$.
	Let $\vec{c}_i$ and $\vec{c}'_i$ be the vector of coefficients of $g_i$ and $g'_i$, respectively, so $\vec{c}_i[p]$ is the coefficient of $\vec{v}[p]$ in $g_i(\vec{v})$ and $\vec{c}'_i[p]$ is the coefficient of $\vec{v}[p]$ in $g'_i(\vec{v})$.
	It suffices to check that $\post{t}\cdot\vec{c}_i - \pre{t}\cdot\vec{c}_i \geq 0$ and $\post{t}\cdot\vec{c}'_i - \pre{t}\cdot\vec{c}'_i \geq 0$ (here $\vec{x}\cdot\vec{y}$ is the dot product of vectors $\vec{x}$ and $\vec{y}$).
	It could be that a place that is not being produced to, that is being reset, contained a token before firing $t$, in which case the value of some function $g_1, g'_1, \ldots, g_k, g'_k$ may strictly decrease.
	We therefore may as well assume that all places due to be reset are empty before $t$ is fired.

	For example, when firing $t = u_i^\bot$, one token is consumed from $h_i$, this decreases $g_i, g'_i, \ldots, g_k, g'_k$ by $2^{k-i+1}$.
	Then one token is produced to $w_i$, this increases $g_i, g'_i, \ldots, g_k, g'_k$ by $2^{k-i}$, and one token is produces to $v_i$, this increases $g'_i, g_{i+1}, g'_{i+1}, \ldots, g_k, g'_k$ by $2^{k-i}$.
	The values of the other functions $g_1, g'_1, \ldots, g_{i-1}, g'_{i-1}$ do not depend on the number of tokens in $h_i$, $w_i$, $v_i$, or $\overline b_i$.
	Overall, without considering the resets, we can conclude that firing $u_i^\bot$ does not incease the value of any function $g_1, g'_1, \ldots, g_k, g'_k$.

	We continue with the proof by splitting into cases $t = u_i^\bot$, $t = u_i^\top$, $t = e_i^\bot$, $t = e_i^\top$, $t \in \set{ \ell_{\overline y_i}, \ell_{y_i}, \ell_{\overline x_i}, \ell_{x_i} }$, and $t = s$.

	\subparagraph*{Case $t = u_i^\bot$:}
	First, $u_i^\bot$ does not consume from, reset, or produce to any place in $U_1, E_1, \ldots, U_{i-1}, E_{i-1}$.
	Since $g_1, g'_1, \ldots, g_{i-1}, g'_{i-1}$ only depend on the contents of places in $U_1, E_1, \ldots, U_{i-1}, E_{i-1}$, it follows immediately that $g_j(\vec{p}) = g_j(\vec{q})$ and $g'_j(\vec{p}) = g'_j(\vec{q})$ for all $j \in [1, i-1]$.

	Now, let us consider the positive contributions that $u_i^\bot$ has on the remaining functions  $g_i, g'_i, \ldots, g_k, g'_k$. 
	For this, we care for what $u_i^\bot$ produces: $\post{{u_i^\bot}}[w_i] = 1$, $\post{{u_i^\bot}}[v_i] = 1$, and $\post{{u_i^\bot}}[\overline b_i] = 2^{k-i}$.
	In $g_i$, $w_i$ contributes $2^{k-i}$ times the number of its tokens and $\overline b_i$ contributes just the number of its tokens, so firing $u_i^\bot$ seemingly has the potential to increase $g_i$ by $2^{k-i} + 2^{k-i} = 2^{k-i+1}$.
	Similarly, in each $g'_i, g_{i+1}, g'_{i+1}, \ldots, g_k, g'_k$, both $w_i$ and $v_i$ contribute $2^{k-i}$ times the number of their tokens, so firing $u_i^\bot$ seemingly has the potential to increase $g'_i, g_{i+1}, g'_{i+1}, \ldots, g_k, g'_k$ by $2^{k-i} + 2^{k-i} = 2^{k-i+1}$.
	However, in order to fire $u_i^\bot$, a token is consumed from $h_i$.
	Observe that $h_i$ contributes $2^{k-i+1}$ times the number of its tokens to all functions $g_i, g'_i, \ldots, g_k, g'_k$. 
	Therefore, the net maximum effect of $u_i^\bot$ on any of $g_i, g'_i, \ldots g_k, g'_k$ is 0.
	We can conclude that $g_j(\vec{p}) \geq g_j(\vec{q})$ and $g'_j(\vec{p}) \geq g'_j(\vec{q})$ for the remaining $j \in [i, k]$.

	\subparagraph*{Case $t = u_i^\top$:} is analogous to case $t = u_i^\bot$, above.

	\subparagraph*{Case $t = e_i^\bot$:}
	Similar to case $t = u_i^\bot$, $e_i^\bot$ does not consume from, reset, or produce to any place in $U_1, E_1, \ldots, U_{i-1}, E_{i-1}, U_i$.
	Since $g_1, g'_1, \ldots, g_{i-1}, g'_{i-1}, g_i$ only depend on the contents of places in $U_1, E_1, \ldots, U_{i-1}, E_{i-1}, U_i$, it follows immediately that $g_j(\vec{p}) = g_j(\vec{q})$ for all $j \in [1, i]$ and $g'_j(\vec{p}) = g'_j(\vec{q})$ for all $j \in [1, i-1]$.

	Again, consider the positive contributions that $e_i^\bot$ has on the remaining functions $g'_i, g_{i+1}, g'_{i+1}, \ldots, g_k, g'_k$.
	The productions are $\post{{e_i^\bot}}[\overline a_i] = 2^{k-i}$ and $\post{{e_i^\bot}}[h_{i+1}] = 1$.
	In $g'_i$, $\overline a_i$ contributes the number of its tokens, so firing $e_i^\bot$ seemingly has the potential to increase $g'_i$ by $2^{k-i}$.
	In each $g_{i+1}, g'_{i+1}, \ldots, g_k, g'_k$, $h_{i+1}$ contributes $2^{k-i}$ times the number of its tokens, so firing $e_i^\bot$ seemingly has the potential to increase $g_{i+1}, g'_{i+1}, \ldots, g_k, g'_k$ by $2^{k-i}$.
	However, in order to fire $u_i^\bot$, a token is consumed from $v_i$.
	Observe that $v_i$ contributes $2^{k-1}$ times the number of its tokens to all functions $g'_i, g_{i+1}, g'_{i+1}, \ldots, g_k, g'_k$.
	Therefore, the net maximum effect of $e_i^\bot$ on any of $g'_i, g_{i+1}, g'_{i+1}, \ldots, g_k, g'_k$ is 0.
	We can conclude that $g_j(\vec{p}) \geq g_j(\vec{q})$ for the remaining $j \in [i+1, k]$ and $g'_j(\vec{p}) \geq g'_j(\vec{q})$ for the remaining $j \in [i, k]$.

	\subparagraph*{Case $t = e_i^\top$:} is analogous to case $t = e_i^\bot$, above.

	\subparagraph*{Case $t = \ell$:} Here $\ell \in \set{ \ell_{\overline y_i}, \ell_{y_i}, \ell_{\overline x_i}, \ell_{x_i} }$.
	We will consider the case when $t = \ell_{y_i}$ and the other cases follow almost identically.

	First, observe that $\ell_{y_i}$ consumes a token from $b_i$ and produces a token to all clauses places whose clause contains $y_i$ (including the dummy clause $\overline y_i \vee y_i$).
	Of all function $g_1, g'_1, \ldots, g_k, g'_k$, only $g_i$ depends on the consumption and production places $\ell_{y_i}$.
	Therefore $g_j(\vec{p}) = g_j(\vec{q})$ for all $j \in [1, k] \setminus\set{i}$ and $g'_j(\vec{p}) = g'_j(\vec{q})$ for all $j \in [1, k]$.
	In this case, $d_{y_i}$ contributes the number of its tokens to $g_i$, so firing $\ell_{y_i}$ seemingly has the potential to increase $g_i$ by 1.
	However, in order to fire $\ell_{y_i}$, a token is consumed from $b_i$ and $b_i$ contributes the number of its tokens to $g_i$.
	Thus, the next maximum effect that firing $\ell_{y_i}$ has on $g_i$ is 0.
	We can finish this case since $g_i(\vec{p}) \geq g_i(\vec{q})$.

	\subparagraph*{Case $t = s$:} 
	In this case, the firing of $s$ has an effect on all functions $g_1, g'_1, \ldots, g_k, g'_k$.
	That is because $s$ consumes from each of the dummy clause places $d_{y_i}$ and $d_{x_i}$, and produces to the final place $f$.
	The positive contribution to all functions is just 1, since $f$ contributes the number of its tokens to all the functions.
	However, to fire $s$, a token is indeed consumed from each of the dummy clause places. 
	Notice that each function depends on one of the dummy clause places' contents: $d_{y_i}$ contributes the number of its tokens to $g_i$ and $d_{x_i}$ contributes the number of its tokens to $g'_i$.
	So, in all $g_1, g'_1, \ldots, g_k, g'_k$, the net maximum effect of firing $s$ is 0.
	Hence, $g_i(\vec{p}) \geq g_i(\vec{q})$ and $g'_i(\vec{p}) \geq g'_i(\vec{q})$ for all $i \in [1, k]$. 
\end{proof}

\goodinvariant*
\begin{proof}
	For sake of contradiction, let us assume that $\run{\vec{p}}{t}{\vec{q}}$, where $\vec{p}$ is a bad marking that is reachable from $\vec{m}$, and $\vec{q}$ is a good marking (that is also reachable from $\vec{m}$).

	There are two ways in which $\vec{p}$ can be a bad marking.
	For some $i \in [1,k], g_i(\vec{v}) \neq 2^k$, or $g'_i(\vec{v}) \neq 2^k$.

	Suppose there exists an $i \in [1, k]$ such that $g_i(\vec{p}) \neq 2^k$.
	Since the initial marking $\vec{m}$ is good, we know that $g_i(\vec{m}) = 2^k$, and by~\cref{clm:invariant-cant-increase}, it must be the case that $g_i(\vec{p}) < 2^k$.
	Again by~\cref{clm:invariant-cant-increase}, this implies that $g_i(\vec{q}) \leq g_i(\vec{p}) < 2^k$.
	Therefore, $\vec{q}$ cannot be good.
\end{proof}

\reachtarget*
\begin{proof}
	First, observe that that $g_i(\vec{m}) = 2^k$ and $g'_i(\vec{m}) = 2^k$ for all $i \in [1, k]$.
	By~\cref{clm:invariant-cant-increase}, all markings $\vec{v}$ seen throughout $\run{\vec{m}}{\pi}{\vec{n}'}$ must have $g_i(\vec{v}) \leq 2^k$ and $g'_i(\vec{v}) \leq 2^k$ for all $i \in [1, k]$.
	So for the final marking $\vec{n}'$, since $\vec{n}' \geq \vec{n}$ it is true that $\vec{n}'[f] \geq \vec{n}[f]$, and because $\vec{n}[f] = 2^k$ and $\vec{n}[f] \leq g_1(\vec{n}') \leq 2^k$, it must be the case that $\vec{n}'[f] = \vec{n}[f] = 2^k$.

	Consider, for any $i \in [1, k]$, the value of $g_i(\vec{n}') - \vec{n}'[f]$ and $g'_i(\vec{n}') - \vec{n}'[f]$.
	By the argument above, $0 \leq g_i(\vec{n}') \leq 2^k$ and $0 \leq g'_i(\vec{n}') \leq 2^k$.
	Given that $\vec{n}'[f] = 2^k$, it must be the case that $g_i(\vec{n}') - \vec{n}'[f] = 0$ and $g'_i(\vec{n}') - \vec{n}'[f] = 0$. 
	This means that all places in $U_1, E_1, \ldots, U_k, E_k$ are empty, including in particular the dummy clause places.
	This implies that the last transition fired in $\pi$ was $s$, because all other transitions produce a token to a place in $U_1, E_1, \ldots, U_k, E_k$.
	In turn, given that $s$ resets all clause places, they must be empty in the marking $\vec{n}'$, so $\vec{n}'[p] = \vec{n}[p] = 0$ for all $p \in P \setminus \set{f}$.
	Altogether, this yields $\vec{n}' = \vec{n}$.
\end{proof}

\badtransitions*
\begin{proof}
	We may assume that $\vec{p}$ is good, for if not, we can immediately conclude that $\vec{q}$ is bad by~\cref{clm:good-invariant}.	
	We split the proof into cases depending on~$t$.
	We focus on the case $t = u_i^\bot$; the other three cases $t = u_i^\top$, $t = e_i^\bot$, and $t = e_i^\top$ follow analogously.

	First suppose there exists $p \in U_j$ (for some $j \in [i, k]$) such that $u_i^\bot \prec p$ and $\vec{p}[p] \geq 1$.
	We show that $g_j(\vec{p}) > g_j(\vec{q})$, implying that $\vec{q}$ is bad.
	By~\cref{clm:invariant-cant-increase}, the maximum net effect on $g_j$ of firing $u_i^\bot$ is~$0$.
	Notice that, in fact, zero effect on $g_j$ can only be achieved when the negative contributions to $g_j$ sum to~$2^{k-i+1}$.
	This is already achieved by the consumption of a token from $h_i$, which indeed contributes $2^{k-i+1}$ times the number of its tokens to~$g_j$.
	The transition $u_i^\bot$ resets all later occurring places, including~$p$.
	Furthermore, $p$ contributes \emph{at least} the number of its tokens, $\vec p[p] > 0$, to~$g_j(\vec p)$ (the exact contribution depends on which place $p$~is).
	Therefore, $g_j(\vec{p}) > g_j(\vec{q})$, as desired.

    The scenario in which there exists $p \in E_j$ (for some $j \in [i, k]$) such that $u_i^\bot \prec p$ and $\vec{p}[p] \geq 1$ is handled in the same way, except that $g_j$ is replaced by $g'_j$ and $2^{k-i+1}$ by $2^{k-i}$.
\end{proof}

\qbfinduction*
\begin{proof}
	We will prove this lemma by induction on $i$.

	First, given~\eqref{enough-fuel}, we know that $(\beta_1, \alpha_1, \ldots, \beta_{k-i}, \alpha_{k-i}) \in \set{0,1}^{2(k-i)}$.
	In other words, we know that initially $\beta_j, \alpha_j \in \{0, 1\}$ for all $j \in [1, k-i]$.

	Recall that, by~\cref{clm:good-invariant}, if a bad marking is ever observed, then $\vec{q}$ cannot be reached. 
	That is because, since $\vec{q}$ can cover $\vec{n}$, by~\cref{clm:reach-target} $\vec{q}$ can only reach $\vec{n}$.
        Since $\vec{n}$ is a good marking, $\vec{q}$ must be a good marking too.

	\subparagraph*{Base case $i=0$:}
	In this case we know that all variables have been assigned their value.
	We would like to argue that the quantifier-free propositional sentence $\phi(\beta_1, \alpha_1, \ldots, \beta_k, \alpha_k)$ evaluates to true and that none of the universal or existential control transitions are fired.
	More precisely, we will show that only one firing sequence can occur: $\sigma = (\ell_1, \ell'_1, \ldots, \ell_k, \ell'_k, s)$ where $\ell_j \in \set{ \ell_{\overline y_j}, \ell_{y_j} }$ and $\ell'_j \in \set{ \ell_{\overline x_j}, \ell_{x_j} }$ for each $j \in [1, k]$.

	Initially, in $\vec{p}$, no universal or existential control transition can be fired.
	Indeed, by~\cref{clm:true-xor-false}, either $\overline a_k$ is non-empty or $a_k$ is non-empty.
	Thus, by~\cref{clm:bad-transitions}, firing any universal or existential control transition would lead to a bad marking.
	Furthermore, $s$ cannot be fired because $\vec{p}$ is clause-free.
	Therefore, the loading transitions must be fired first.
	The loading transitions must in fact be fired in sequence: $\ell_{\overline y_1}$ or $\ell_{y_1}$, then $\ell_{\overline x_1}$ or $\ell_{x_1}$, then $\ell_{\overline y_2}$ or $\ell_{y_2}$, and so on until $\ell_{\overline x_k}$ or $\ell_{x_k}$.
	Again, if any are fired out of sequence, then by~\cref{clm:bad-transitions} a bad marking would be reached.

	After this, the only transition that can be fired is $s$, which produces a token to $f$.
    Given that $\vec{q}[f] = \vec{p}[f] + 1$, this transition must indeed be fireable.
	This means that all clause places contain a token after the loading transitions are fired.
	In other words, each clause is satisfied under the current assignment $y_j \leftarrow \beta_j$, $x_j \leftarrow \alpha_j$, $j \in [1, k]$, so $\phi(\beta_1, \alpha_1, \ldots, \beta_k, \alpha_k)$ evaluates to true.
	This must be the final transition of $\sigma$, by conditions~\eqref{div-and-diff} and~\eqref{last-fire}.
	It is therefore clear that $\sigma$ does not fire any of the universal or existential control transitions.

	\subparagraph*{Inductive step $i \to i+1$:}
	In this case we assume the lemma holds for $i$.
	Here, we would like to show that the following QBF evaluates to true.
	\begin{equation*}
		F_{i+1} \coloneqq \forall y_{k-i} \, \exists x_{k-i} \, \ldots \, \forall y_k \, \exists x_k : \phi(\beta_1, \alpha_1, \ldots, \beta_{k-i-1}, \alpha_{k-i-1}, y_{k-i}, x_{k-i}, \ldots, y_k, x_k) 
	\end{equation*}
	We would also like to show that $\sigma$ does not fire $u_j^\bot$, $u_j^\top$, $e_j^\bot$, or $e_j^\top$ for any $j \in [1, k-i-1]$.
	We will in fact show that $\sigma = u_{k-i}^\bot \, e \, \sigma_0 u_{k-i}^\top e' \sigma_1$, where $e, e' \in \set{ e_{k-i}^\bot, e_{k-i}^\top }$ and both $\sigma_0$ and $\sigma_1$ are obtained by two calls to the inductive assumption for~$i$: once when $y_{k-i}$ is set to false, and once when $y_{k-i}$ is set to false. 

	Initially, in $\vec{p}$, none of the earlier universal or existential control transitions can be fired, $u_j^\bot$, $u_j^\top$, $e_j^\bot$, or $e_j^\top$ for $j \in [1, k-i-1]$.
	That is because $\vec{p}[h_{k-i}] = 1$, and by~\cref{clm:bad-transitions}, this would lead to a bad marking.
	Moreover, none of the \emph{later} universal or existential control transitions $u_{k-i}^\top$, $e_{k-i}^\bot$, $e_{k-i}^\top$, \ldots, $u_k^\bot$, $u_k^\top$, $e_k^\bot$, or $e_k^\bot$ can be fired, because all places $p \in U_{k-i} \cup E_{k-i} \cup \cdots \cup U_k \cup E_k$ except $h_{k-i}$ are empty.
	The loading transitions cannot be fired either, for currently the places $\overline a_k$ and $a_k$ are empty.
	Indeed, one of these places needs to be non-empty later in $\sigma$ so that $d_{x_k}$ can be non-empty, which is necessary for the firing of $s$.
	To eventually make $\overline a_k$ or $a_k$ non-empty, $e_k^\bot$ or $e_k^\top$ needs to be fired.
	When $e_k^\bot$ or $e_k^\top$ are fired, all dummy clause places are reset.
	Hence if a loading transition is prematurely fired, by claim~\cref{clm:bad-transitions}, the firing of $e_k^\bot$ or $e_k^\top$ would lead to a bad marking.
	Finally, $s$ cannot be fired because $\vec{p}$ is clause-free.
	The only transition that remains is $u_{k-i}^\bot$: $\sigma$ first fires $u_{k-i}^\bot$.

	Let $\vec{p}'$ be the marking reached after firing $u_{k-i}^\bot$.
	Now $\vec{p}'[h_{k-i}] = 0$, $\vec{p}'[w_{k-i}] = 1$, $\vec{p}'[v_{k-i}] = 1$, and $\vec{p}'[\overline b_{k-i}] = 2^{i-1}$.
	For much like the above, none of the earlier universal or existential control transitions can be fired, $u_j^\bot$, $u_j^\top$, $e_j^\bot$, or $e_j^\top$ for $j \in [1, k-i-1]$, given that now $\vec{p}[w_{k-i}] = 1$, and by~\cref{clm:bad-transitions}, this would lead to a bad marking.
	Additionally, $u_{k-i}^\bot$ cannot be fired now that $\vec{p}'[h_{k-i}] = 0$, and $u_{k-i}^\top$ cannot be fired now that $\overline b_{k-i}$ contains tokens (again, its firing would lead to a bad marking by~\cref{clm:bad-transitions}).
	Moreover, none of the \emph{later} universal or existential control transitions $u_j^\bot$, $u_j^\top$, $e_j^\bot$, or $e_j^\top$ for $j \in [i+2, k]$ can be fired, because all places $p \in U_{k-i} \cup E_{k-i} \cup \cdots \cup U_k \cup E_k$ such that $w_{k-i} \prec p$ are empty.
	Similarly, none of the loading transitions can be fired.
	The satisfaction transition~$s$ cannot be fired because $\vec{p}'$ is clause-free.
	The only transitions remaining are $e_{k-i}^\bot$ and $e_{k-i}^\top$. 

	We will denote by $\alpha \in \set{0,1}$ the value assigned to~$x_{k-i}$, determined as follows.
	If $\sigma$ next fires $e_{k-i}^\bot$, then $\alpha = 0$; otherwise $\sigma$ next fires $e_{k-i}^\top$, and then $\alpha = 1$.
	We will use the inductive assumption (for~$i$) to show that $F_{i+1}$ evaluates to true when $y_{k-i} \leftarrow 0$ and $x_{k-i} \leftarrow \alpha$.
	Let $\vec{p}''$ be the marking reached after firing $e \in \set{ e_{k-i}^\bot, e_{k-i}^\top }$.
	Since $e_{k-i}^\bot$ produces $2^i$ many tokens to $\overline a_{k-i}$ and $e_{k-i}^\top$ produces $2^i$ many times to $a_{k-i}$, we know that $\vec{p}''[\overline a_{k-i}] + \vec{p}''[a_{k-i}] \geq 2^i$.
	Recall that the earlier firing of $u_{k-i}^\bot$ produced $2^i$ many tokens to $\overline b_{k-i}$, so $\vec{p}''[\overline b_{k-i}] + \vec{p}[b_{k-i}] \geq 2^i$ also holds.
	Firing $e$ also places a token in $h_{k-i+1}$ and resets all later places $p \in U_{k-i+1} \cup E_{k-i+1} \cup \cdots \cup U_k \cup E_k$, i.e., those with $h_{k-i+1} \prec p$.
	The marking $\vec{p}''$ can still cover $\vec{n}$, is clearly reachable from $\vec{p}$ (so reachable from $\vec{m}$), and is clause-free because neither $u_{k-i}^\bot$ nor $e$ produce a token to a clause place.
	We call upon the inductive hypothesis for~$i$ on a firing sequence $\sigma_0$ for a run $\run{\vec{p}''}{\sigma_0}{\vec{r}}$, where $\vec r$ is the marking reached by $\sigma$ immediately after $2^i$ firings of~$s$.
	We know that $\sigma_0$ does not fire $u_j^\bot$, $u_j^\top$, $e_j^\bot$, or $e_j^\top$, for any $j \in [1, k-i]$, and we know that the following partial QBF evaluates to true.
	\begin{equation}
		\forall y_{k-i+1} \, \exists x_{k-i+1} \, \ldots, \, \forall y_k \, \exists x_k : \phi(\beta_1, \alpha_1, \ldots, \beta_{k-i-1}, \alpha_{k-i-1}, 0, \alpha, y_{k-i+1}, x_{k-i+1}, \ldots, y_k, x_k)
		\label{equ:qbf-one}
	\end{equation}
        In order for $\sigma_0$ to fire~$s$ exactly $2^i$ times, both $d_{y_{k-i}}$ and $d_{x_{k-i}}$ must have been non-empty $2^i$ many times, and this requires firing $\ell_{\overline y_{k-i}}$ and $\ell'$ (where $\ell' \in \set{ \ell_{\overline x_{k-i}}, \ell_{x_{k-i}} }$) $2^i$ times each.
	Since $\vec{p}''[\overline b_{k-i}] = 2^i$, $\vec{p}''[b_{k-i}] = 0$, and $\vec{p''}[\overline a_{k-i}] + \vec{p''}[a_{k-i}] = 2^i$, we know that $\vec{r}[\overline b_{k-i}] = \vec{r}[b_{k-i}] = \vec{r}[\overline a_{k-i}] = \vec{r}[a_{k-i}] = 0$.
	We also note that $\vec{r}$ is clause-free, since the final transition fired in $\sigma_0$ is $s$.

	Let us now continue, by considering the next transition that can be fired by $\sigma$.
	The analysis is almost identical; the main difference is that $\vec{r}[h_{k-i}] = 0$ and $\vec{r}[w_{k-i}] = 1$ (as opposed to $\vec{p}[h_{k-i}] = 1$ and $\vec{p}[w_{k-i}] = 0$).
	This means that the first transition that can be fired is $u_{k-i}^\top$.
	Firing this transition places $2^i$ many tokens in $b_{k-i}$, and one token in $v_{k-i}$, at an intermediate, still clause-free, marking $\vec{r}'$.
	Following this, again for the same reasons as before, the only transitions that can be fired are $e_{k-i}^\bot$ and $e_{k-i}^\top$.

	Let $\alpha' \in \set{0,1}$ be the next assignment that $x_{k-i}$ will receive.
	Just like before, if $\sigma$ next fires $e_{k-i}^\bot$, then $\alpha' = 0$; otherwise $\sigma$ next fires $e_{k-i}^\top$, and then $\alpha' = 1$.
	We will again use the inductive hypothesis for~$i$ to show that $F_{i+1}$ evaluates to true when $y_{k-i} \leftarrow 1$ and $x_{k-i} \leftarrow \alpha'$.
	Indeed, let $\vec{r}''$ be the marking reached after firing $e' \in \set{ e_{k-i}^\bot, e_{k-i}^\top }$; $\vec{r}''[\overline a_{k-i}] + \vec{p}''[a_{k-i}] \geq 2^i$.
	Recall that the earlier firing of $u_{k-i}^\top$ produced $2^i$ many tokens to $b_{k-i}$, so $\vec{p}''[\overline b_{k-i}] + \vec{p}[b_{k-i}] \geq 2^i$ also holds.
	Firing $e'$ also places a token in $h_{k-i+1}$ and resets all later places $p \in U_{k-i+1} \cup E_{k-i+1} \cup \cdots \cup U_k \cup E_k$, i.e., those with $h_{k-i+1} \prec p$.
	The marking $\vec{r}''$ can still cover $\vec{n}$ and is clearly reachable from $\vec{r}$, which in turn is reachable from~$\vec{p}$.
        So $\vec{r}''$ is reachable from $\vec{m}$ too.
	It is also a clause-free marking because neither $u_{k-i}^\top$ nor $e'$ produce a token to a clause place.
	We call upon the inductive hypothesis for~$i$ on a firing sequence $\sigma_1$ for a run $\run{\vec{r}''}{\sigma_1}{\vec{q}}$.
	We know that $\sigma_1$ does not fire $u_j^\bot$, $u_j^\top$, $e_j^\bot$, or $e_j^\top$, for any $j \in [1, k-i]$, and we know that the following partial QBF evaluates to true.
	\begin{equation}
		\forall y_{k-i+1} \, \exists x_{k-i+1} \, \ldots, \, \forall y_k \, \exists x_k : \phi(\beta_1, \alpha_1, \ldots, \beta_{k-i-1}, \alpha_{k-i-1}, 1, \alpha', y_{k-i+1}, x_{k-i+1}, \ldots, y_k, x_k)
		\label{equ:qbf-two}
	\end{equation}

	We can now can bring together~\cref{equ:qbf-one} and~\cref{equ:qbf-two} to deduce that the following QBF evaluates to true:
	\begin{equation*}
		\forall y_{k-i} \, \exists x_{k-i} \, \ldots, \, \forall y_k \, \exists x_k : \phi(\beta_1, \alpha_1, \ldots, \beta_{k-i-1}, \alpha_{k-i-1}, y_{k-i}, x_{k-i}, \ldots, y_k, x_k).
	\end{equation*}
	Finally, we know that $\sigma = (u_{k-i}^\bot, e, \sigma_0, u_{k-i}^\top, e', \sigma_1)$ does not fire $u_j^\bot$, $u_j^\top$, $e_j^\bot$, or $e_j^\top$ for any $j \in [1, k-i-1]$ because $\sigma_0$ and $\sigma_1$ do themselves not fire $u_j^\bot$, $u_j^\top$, $e_j^\bot$, or $e_j^\top$ for any $j \in [1, k-i]$, and beyond that $\sigma$ fires only $u_{k-i}^\bot$, $u_{k-i}^\top$, and at least one of $e_{k-i}^\bot$ and $e_{k-i}^\top$.
    This completes the proof.
\end{proof}

\coverabilityrun*
\begin{proof}
	We will prove this lemma by induction on $i$.
	
	\subparagraph{Base case $i = 0$:}
	Given $\beta_1, \alpha_1, \ldots, \beta_k, \alpha_k \in \set{0, 1}$ we will define the firing sequence $\sigma = (\ell_1, \ell'_1, \ldots, \ell_k, \ell'_k, s)$ where, for every $j \in [1, k]$,
	\begin{equation*}
		\ell_j = 
		\begin{cases}
			\ell_{\overline y_j} & \text{if } \beta_j = 0 \\
			\ell_{y_j} & \text{if } \beta_j = 1
		\end{cases}
                \qquad
		\text{and}
                \qquad
		\ell'_j = 
		\begin{cases}
			\ell_{\overline x_j} & \text{if } \alpha_j = 0 \\
			\ell_{x_j} & \text{if } \alpha_j = 1.
		\end{cases}
	\end{equation*}

	From $\vec{p}$, the transition $\ell_1$ can be fired, since if, for example, $\beta_1 = 0$, then $\ell_1 = \ell_{\overline y_1}$.
	We know that $\ell_{\overline y_1}$ can be fired because, by (1), $\overline b_1$ contains a token.
	Firing $\ell_1$ does not reset any of the later non-clause places, so $\ell'_1$ can be fired.
	This is for the same reason, for example if $\alpha_1 = 1$, then $\ell'_1 = \ell_{x_1}$ can be fired because, by (2), $a_1$ contains a token.
	This argument repeats for the remaining loading transitions $\ell_2, \ell'_2, \ldots, \ell_k, \ell'_k$.

	Suppose	$\vec{p}\xrightarrow{\ell_1 \ell'_1 \cdots \ell_k \ell'_k}\vec{r}$, we will now argue that $\vec{r}[c] \geq 1$ for all $c \in C$.
	First, it is clear that all dummy clause places are non-empty: for all $j \in [1, k]$ either $\ell_{\overline y_j}$ or $\ell_{y_j}$ was fired so $\vec{r}[d_{y_j}] = 1$, and either $\ell_{\overline x_j}$ or $\ell_{x_j}$ was fired, so $\vec{r}[d_{x_j}] = 1$.
	Second, we know that $\phi(\beta_1, \alpha_1, \ldots, \beta_k, \alpha_k)$ evaluates to true, so every clause of $\phi$ contains a true literal.
	A loading transition is fired if its literal is true, and produces a token to all clause places that contain said literal.
	Therefore, after firing all loading transitions, all clause places must contain a token.

	Finally, consider the marking $\vec{q}$ reached after firing $s$ from $\vec{r}$ (that transition is fireable given that $\vec{r}[c] \geq 1$ for all $c \in C$).
	Indeed, given that $s$ produces a token to $f$, $\vec{q}[f] = \vec{r}[p] + 1$.
	All together, $\run{\vec{p}}{\sigma}{\vec{q}}$ where $\vec{q}[f] = \vec{r}[f] + 1 = \vec{p}[f] + 1$.
	One of each of the two loading transitions for a variable was fired just once, so both (a) and (b) hold.
	No transitions in $\sigma$ consumed from or reset any of the holding places, waiting places, or decision places.
	Thus, $\vec{q}[p] = \vec{p}[p]$ for all $p \in \set{h_1, w_1, v_1, \ldots, h_k, w_k, v_k}$, so (c) holds.

	\subparagraph{Inductive step $i \to i+1$:}
	In this case we shall assume that the lemma holds for case $i$. 
	Given $\beta_1, \alpha_1, \ldots, \beta_{k-i-1}, \alpha_{k-i-1} \in \set{0, 1}$, we will define the firing sequence $\sigma = (u_{k-i}^\bot, e, \sigma_0, u_{k-i}^\top, e', \sigma_1)$ where $e$ and $e'$ are selected to set $x_{k-i}$ to true or false depending on the assignment of $y_{k-i}$ and where $\sigma_0$ and $\sigma_1$ are firing sequences given by case $i$:
	\begin{equation*}
		e = 
		\begin{cases}
			e_{k-i}^\bot & \text{if } \forall y_{k-i+1} \, \exists x_{k-i+1} \, \ldots \, \forall y_k \, \exists x_k : \\
			& \phi(\beta_1, \alpha_1, \ldots, \beta_{k-i-1}, \alpha_{k-i-1}, 0, 0, y_{k-i+1}, x_{k-i+1}, \ldots, y_k, x_k) \text{ evaluates to true,}\\
			e_{k-i}^\top & \text{otherwise, i.e., if } \forall y_{k-i+1} \, \exists x_{k-i+1} \, \ldots \, \forall y_k \, \exists x_k : \\ 
			& \phi(\beta_1, \alpha_1, \ldots, \beta_{k-i-1}, \alpha_{k-i-1}, 0, 1, y_{k-i+1}, x_{k-i+1}, \ldots, y_k, x_k) \text{ evaluates to true;}
		\end{cases}
	\end{equation*}
	\begin{equation*}
		e' = 
		\begin{cases}
			e_{k-i}^\bot & \text{if } \forall y_{k-i+1} \, \exists x_{k-i+1} \, \ldots \, \forall y_k \, \exists x_k : \\
			& \phi(\beta_1, \alpha_1, \ldots, \beta_{k-i-1}, \alpha_{k-i-1}, 1, 0, y_{k-i+1}, x_{k-i+1}, \ldots, y_k, x_k) \text{ evaluates to true,}\\
			e_{k-i}^\top & \text{otherwise, i.e., if } \forall y_{k-i+1} \, \exists x_{k-i+1} \, \ldots \, \forall y_k \, \exists x_k : \\ 
			& \phi(\beta_1, \alpha_1, \ldots, \beta_{k-i-1}, \alpha_{k-i-1}, 1, 1, y_{k-i+1}, x_{k-i+1}, \ldots, y_k, x_k) \text{ evaluates to true.}
		\end{cases}
	\end{equation*}

	Initially, given that $\vec{p}[h_{k-i}] = 1$, we know $u_{k-i}^\bot$ can be fired, and suppose $\run{\vec{p}}{u_{k-i}^\bot}{\vec{p}'}$.
	Now $\vec{p'}[v_{k-i}] = 1$ so $e$ can be fired.
	Moreover, $\vec{p'}[w_{k-i}] = 1$ and $\vec{p'}[\overline b_{k-i}] = 2^i$.
	Suppose $\run{\vec{p}'}{e}{\vec{p}''}$.
	Now $\vec{p''}[h_{k-i+1}] = 1$ and since $e$ did not reset $\overline b_{k-i}$, we know that both $\vec{p''}[\overline b_{k-i}] = 2^i$ and either $\vec{p''}[\overline a_{k-i}] = 2^i$ or $\vec{p''}[a_{k-i}] = 2^i$.
	We can now call upon the inductive hypothesis for~$i$ to obtain a firing sequence $\sigma_0$ such that $\run{\vec{p}''}{\sigma_0}{\vec{r}}$.
	
	We know that $\vec{r}[f] = \vec{p}''[f] + 2^i$.
	Conditions (a) and (b) yield $\vec{r}[\overline b_{k-i}] = 0$ and $\vec{r}[\overline a_{k-i}] = \vec{r}[a_{k-i}] = 0$.
	These conditions also tell us that, for every $j \in [1, k]$, $\vec{r}[\overline b_j] + \vec{r}[b_j] = \vec{p''}[\overline b_j] + \vec{p''}[b_j] - 2^i = \vec{p}[\overline b_j] + \vec{p}[b_j] - 2^i$ and $\vec{r}[\overline a_j] + \vec{r}[a_j] = \vec{p''}[\overline a_j] + \vec{p''}[a_j] - 2^i = \vec{p}[\overline a_j] + \vec{p}[a_j] - 2^i$, so still $\vec{r}[\overline b_j] \geq 2^i$ or $\vec{r}[b_j] \geq 2^i$, and $\vec{r}[\overline a_j] \geq 2^i$ or $\vec{r}[a_j] \geq 2^i$.
	Therefore, $\vec{r}$ will be ready to satisfy (1) and (2) in later call to the inductive hypothesis for~$i$.

	We continue by firing $u_{k-i}^\top$ from $\vec{r}$.
	We can fire $u_{k-i}^\top$ from $\vec{r}$ because, with condition (c) and $\vec{p}''[w_{k-i}] = 1$, we know that $\vec{v}[w_{k-i}] = 1$.
	Suppose $\run{\vec{r}}{u_{k-i}^\top}{\vec{r}'}$.
	Now $\vec{r}'[v_{k-i}] = 1$ so $e'$ can be fired.
	Moreover, $\vec{p'}[b_{k-i}] = 2^i$.
	Suppose $\run{\vec{r}'}{e'}{\vec{r}''}$.
	Now $\vec{r}''[h_{k-i+1}] = 1$ and since $e'$ did not reset $b_{k-i}$, we know that both $\vec{r''}[b_{k-i}] = 2^i$ and either $\vec{r''}[\overline a_{k-i}] = 2^i$ or $\vec{r''}[a_{k-i}] = 2^i$.
	We can now call upon the inductive hypothesis for~$i$ to obtain a firing sequence $\sigma_1$ such that $\run{\vec{r}''}{\sigma_1}{\vec{q}}$.

	We know that $\vec{q}[f] = \vec{r}''[f] + 2^i = \vec{r}[f] + 2^i = \vec{p''}[f] + 2^{i-1} + 2^i = \vec{p} + 2^{i+1}$, as required.
	From (a) and (b) of $\sigma_1$, we know that another $2^i$ tokens were consumed from $\overline b_1$ or $b_1$, $\overline a_1$ or $a_1$, $\ldots$, $\overline b_{k-i}$ or $b_{k-i-1}$, $\overline a_{k-i-1}$ or $a_{k-i-1}$, so given $\sigma_0$ did the same, both (a) and (b) hold for $\sigma$ overall.
	Lastly, $\sigma$ does not consume or reset any of the place $h_j$, $w_j$, or $v_j$ for all $j \in [1, k-i-1]$, so (c) holds.
\end{proof}
\end{appendix}

\end{document}